\documentclass[11pt]{article}
 
\usepackage[a4paper,
            bindingoffset=0.2in,
            left=1in,
            right=1in,
            top=1in,
            bottom=1in,
            footskip=.25in]{geometry}
\usepackage[utf8]{inputenc}
\usepackage{hyperref}
\usepackage[normalem]{ulem}
\usepackage[utf8]{inputenc}
\usepackage[english]{babel}

\usepackage{amsmath, amssymb, amsfonts, amsthm}
\usepackage{graphicx}
\usepackage{xcolor}
\usepackage{booktabs}
\usepackage{algorithm2e}
\usepackage{caption}
\usepackage{subcaption}

\graphicspath{{figures_small/}}

\usepackage{float}
\newfloat{algorithm}{t}{lop}
\floatname{algorithm}{Algorithm}

\SetKwComment{Comment}{/* }{ */}

\definecolor{ultramarine}{rgb}{0.07, 0.04, 0.56}
\definecolor{darkspringgreen}{rgb}{0.09, 0.45, 0.27}
\hypersetup{colorlinks,linkcolor={darkspringgreen},citecolor={ultramarine},urlcolor={ultramarine}}
\newcommand{\tab}[2]{
  
  \begin{table}[htb]
    \centering
    \input{tables/#1}
    \caption[\protect\detokenize{#1}]{#2}
    \label{tab:#1}
  \end{table}

}
\newcommand{\tabf}[2]{
  
  \begin{table}[p]
    \centering
    \input{tables/#1}
    \caption[\protect\detokenize{#1}]{#2}
    \label{tab:#1}
  \end{table}

}

\newcommand{\refsec}[1]{Section~\ref{#1}}
\newcommand{\reftab}[1]{Table~\ref{tab:#1}}

\newcommand{\reffig}[1]{Figure~\ref{#1}}

\def\fixaccepted#1{{\color{magenta} #1}}

\definecolor{EW}{HTML}{E41A1C}
\definecolor{GMV}{HTML}{377EB8}
\definecolor{GMV_lin}{HTML}{4DAF4A}
\definecolor{GMV_long}{HTML}{984EA3}
\definecolor{GMV_nlin}{HTML}{FF7F00}
\definecolor{GMV_{robust}}{HTML}{A65628}
\definecolor{Index}{HTML}{999999}

\newcommand\myeq{\overset{\mathrm{def}}{=}}

\newcommand{\quantlet}[2]{\hspace*{\fill}
  \raisebox{-1pt}{\includegraphics[scale=0.05]{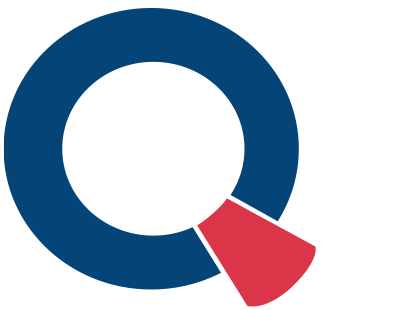}
    \href{https://github.com/QuantLet/RobustM/tree/master/#1}{#2}}
}

\title{Robustifying Markowitz}
\author{Wolfgang Karl H\"ardle\footnote{\scriptsize Blockchain Research Center, Humboldt-Universit\"at zu Berlin, Germany.  Wang Yanan Institute for Studies in Economics, Xiamen University, China.  Sim Kee Boon Institute for Financial Economics, Singapore Management University and Asia Competitiveness Institute (ACI), National University of Singapore (NUS),  Singapore. Faculty of Mathematics and Physics, Charles University, Czech Republic.   Yushan Scholar 
National Yangming Jiaotong University, Taiwan.
\url{haerdle@hu-berlin.de}},
Yegor Klochkov\footnote{\scriptsize ByteDance AI Lab. This work was done when this author was at Janeway Institute, Faculty of Economics, University of Cambridge \url{yk376@cam.ac.uk}}, Alla Petukhina\footnote{\scriptsize Hochschule f\"ur Technik und Wirtschaft Berlin, School of Computer Science, Communication and Economics   \url{alla.petukhina@htw-berlin.de}}, Nikita Zhivotovskiy\footnote{\scriptsize Department of Statistics, UC Berkeley. This work was done when the author was at the Department of Mathematics, ETH Z\"urich,  \url{zhivotovskiy@berkeley.edu}}}
\date{\today}

\usepackage{natbib}
\usepackage{graphicx}
\def\P{\mathsf{P}}
\def\E{\mathsf{E}}
\def\R{\mathbb{R}}
\def\T{\top}
\def\onev{\mathbf{1}}

\newtheorem{theorem}{Theorem}[section]
\newtheorem{proposition}{Proposition}[section]
\newtheorem{lemma}{Lemma}[section]
\newtheorem{corollary}{Corollary}[section]
\newtheorem{assumption}{Assumption}[section]
\newtheorem*{theorem*}{Theorem}

\newtheorem{definition}{Definition}[section]
\newtheorem{remark}{Remark}[section]

\begin{document}
\sloppy
\maketitle

\begin{abstract}
  \small
  Markowitz mean-variance portfolios with sample mean and covariance as input parameters feature numerous issues in practice. They perform poorly out of sample due to estimation error, they experience extreme weights together with high sensitivity to change in input parameters. The heavy-tail characteristics of financial time series are in fact the cause for these erratic fluctuations of weights that consequently create substantial transaction costs. In robustifying the weights we present a toolbox for stabilizing costs and weights for global minimum Markowitz portfolios. Utilizing a projected gradient descent (PGD) technique,  we avoid the estimation and inversion of the covariance operator as a whole and concentrate on robust estimation of the gradient descent increment. Using modern tools of robust statistics we construct a computationally efficient estimator with almost Gaussian properties based on median-of-means uniformly over weights. This robustified Markowitz approach is confirmed by empirical studies on equity markets. We demonstrate that robustified portfolios reach the lowest turnover compared to shrinkage-based and constrained portfolios while preserving or slightly improving out-of-sample performance.
\end{abstract}
\section{Introduction}

The cornerstone mean-variance portfolio theory proposed by \cite{Markowitz1952} plays a significant role in research and practice. Efficient mean-variance portfolios (MV) experience a number of attractive properties and have a simple and  straightforward analytical solution with only two input parameters: the expected mean and covariance matrix of asset returns. 
Mean-variance analysis is naturally connected to the Capital Asset Pricing Model (CAPM), a standard tool in asset pricing.

Despite its simplicity and theoretical appeal, implementation of mean-variance portfolios is often impractical. The traditional approach to use the sample moments as input parameters leads to extreme negative and positive weights, and extensive literature documents poor out-of-sample performance of such plug-in approach, see \citet{frost1986empirical, Frost29, Best1991, Chopra1993, broadie1993computing, litterman2004modern, merton1980estimating}. 
{The problem might be seen as an inverse problem, and it} has high sensitivity to even small perturbations of the input estimates: the mean and {the} covariance matrix. 
It is possibly surprising that the MV portfolios are more 
sensitive to changes in the mean estimate, 
\cite{jagannathan2003risk} spell this out explicitly by writing that the error of mean estimation is so large that nothing is lost when one ignores the mean at all, and \cite{michaud1989markowitz} describes the influence of the mean error as
``error-maximization''. 
{Following the majority of research on this topic, }
we focus here on \emph{global minimum variance} portfolios (GMV), which only depend on the covariance.

However, even with the mean left out of the equation, traditional policies suffer from extreme instability, which means that the portfolio weights fluctuate significantly over time.  Drastic changes in the portfolio composition lead to increasing management and transaction costs and consequently reducing the popularity of MV policies  among investors.
In order to improve upon the stability of portfolio weights, one has to resort to alternative, robust estimation techniques. A \emph{robust estimator} is one that performs well even when the observations do not follow the standard (normality) assumptions, have heavy tails, or are even subject to contamination. Although in case of normal distributions, sample moment estimators are asymptotically optimal MLEs, they are not necessarily the best choice when the data deviates from normality \citet{huber2004robust}. This is of particular importance in financial applications, where it is well known that the data is not only non-Gaussian, but also exhibits heavy tails.

To tackle the problem of heavy tails, \cite{demiguel2009portfolio} construct a portfolio optimization procedure based on M- and S-estimation technique and analyze the stability of the estimator analytically; they also demonstrate empirically that their approach reduces portfolio \emph{turnover}, whereas it slightly improves the out-of-sample performance. \cite{fan2019robust} {construct} an elementwise covariance estimator through an M-estimation procedure with Huber loss, providing statistical high-probability guarantees. 
Robust portfolio optimization problem  has gained significant attention, \cite{xidonas2020robust} categorize 148 {studies} conducted during the
last 25 years and focused on this topic. 

Failures of MV portfolios become even more pronounced with a growing investment universe, especially for cases when a sample size is  less than the number of assets. Evidence was investigated by \cite{ kan2007optimal}, \cite{ Bai2009}, \cite{ Karoui2012}, and \cite{Chen2016}. To overcome this curse of dimensionality, structured covariance matrix estimators are proposed for asset return data.  \cite{fan2008high} considered estimators based on factor models with observable factors. \cite{stock2002forecasting}, \cite{bai2012statistical},  \cite{fan2013large} studied covariance matrix estimators based on latent factor models. \cite{ledoit2003improved}, \cite{ledoit2004well}, \cite{ledoit2004honey} proposed linear and \cite{ledoit2017nonlinear} non-linear shrinkage of sample eigenvalues. These estimators are commonly based on the sample covariance matrix, and sub-Gaussian tail assumptions are required to guarantee sharp non-asymptotic bounds (see Section 4.7 in \citet{vershynin2018high}).

The goal of our robustifying Markowitz approach is stabilizing portfolio weights by solving two problems at the same time: How to optimize the GMV portfolio when the dimension is possibly higher than the sample size and the distribution of the returns has heavy tails? Moreover, even if the returns are not heavy-tailed, how can one avoid the usual Gaussian assumption in the theoretical analysis?

Our theoretical and algorithmic contributions dwell on some recent breakthroughs in statistical literature with regard to robust estimation. \cite{lugosi2019sub} constructed a multivariate mean estimator based on the idea of median-of-means that dates back to \cite{nemirovsky1983problem}. The remarkable property of their estimator is that it pertains favorable properties of the Gaussian sample mean: it allows deviation bounds with high probability without much loss in the accuracy. Their only condition is that the second moment of each component of the random vector is bounded, which is the minimal possible condition to have a square-root convergence, even on average. However, their original estimator was not computationally tractable and in the past years the problem has attracted a lot of attention. \cite{hopkins2018sub} first proposed an estimator with polynomial computational (efficient) complexity, and subsequent research led to nearly linear-time algorithms \citet{depersin2019robust, hopkins2020robust}, thus making practical applications possible. As for the covariance estimation, \cite{mendelson2020robust} proposed an abstract algorithm that achieves performance of Gaussian sample covariance  estimator under four bounded moments assumption. So far, it remains an open question whether such performance can be achieved with a polynomial algorithm, with some conjecturing that the answer is no, \citet{cherapanamjeri2020algorithms}. 

We bypass these algorithmic problems appearing in robust estimation of the covariance matrix. In fact, our approach does not require estimating the covariance operator directly. It is based on a simple iterative gradient descent, that requires estimating only the action of the covariance operator on a current approximation at each step.

Our contribution to robustifying Markowitz is threefold:
\begin{itemize}
\item 
Based on the algorithm from \citet{hopkins2020robust}, we introduce a robust and computationally tractable algorithm that achieves nearly Gaussian performance under only four moments assumption on the distribution of the return vector. This means, in particular, that the estimator works with almost any distribution with four bounded moments as good as it works with Gaussian data. 

\item We provide theoretical guarantees for our method in two cases. In the first case, we assume that the covariance matrix is well-conditioned, which means that the objective of our optimization problem enjoys the strong convexity property, and the convergence guarantees are provided even for the mean-variance objective. However, in that case we require that the dimension of the investment universe ($N$) is much smaller than the size of the sample ($T$). Moreover, the assumption that the covariance is well-conditioned is impractical due to the presence of strong factors in financial panel data, and we provide this result merely out of theoretical curiosity.

In the general case where we have no control over the small eigenvalues of the covariance matrix, we only consider the GMV objective. However, we can take advantage of possibly small \emph{effective rank} of the covariance matrix, which allows the dimension $N$ to be a lot larger than the sample size $T$.

\item In our empirical study, we compare behavior of the proposed portfolio estimator to the traditional portfolio benchmarks on  equity data for two cases: when the size of the sample $T$ is comparable with the dimension $N$, and when $T<N$. For the first case, we consider the S\&P100 data, and for the second case, we take the constituents of the Russell3000 index. In both cases, we conduct study for daily data over the course of 12 months and 24 months, which corresponds to $T = 252$ and $T = 500$. We demonstrate that our approach enjoys more stable weights than the traditional portfolios, while preserving (or slightly improving) their  out-of-sample performance.
\end{itemize}

Let us also recall a well known hypothesis of \cite{Green1992} that extreme portfolio weights appear not entirely due to high estimation errors, but rather due to the population optimal portfolios themselves having extreme weights and being poorly diversified. Specifically, they show that asset returns generated by a model with a single dominant factor result in excessive short and long positions. This {leads} to {the study of} restricted portfolio policies. In a seminal work, \cite{jagannathan2003risk} consider portfolios with non-negative constraints. Despite considering a lesser class of portfolios, they demonstrate a better out-of-sample performance. Furthermore,  \cite{fan2012vast} introduce \emph{Gross Exposure Constraints}, which work similarly but allow negative allocation weights.
Contrary to these ideas, we demonstrate that applying our robust procedure leads to desirable properties of weights without any constraints enforced a priori, which contradicts the original hypothesis of Green and Hollifield.

\def\onev{\mathbf{1}}
\def\Tr{\operatorname{Tr}}

\subsection{Notation}
Throughout the paper we write $ a \lesssim b $  and $ b \gtrsim a $ if there is a constant $C$ such that $ a \leq C b $. If we have both $ a \lesssim b $ and $ b \lesssim a $, we write $ a\sim b $.

For a vector $ x \in \R^{d} $, we denote by $ \| x \| = \sqrt{x_1^2 + \dots + x_d^2} $ its Euclidean norm. If $ A $ is a matrix, we denote $ \| A \| = \sup_{u, v \in \mathbb{S}^{d - 1}} u^{\T} A v $ its spectral norm, where $\mathbb{S}^{d - 1}$ is a sphere in $\R^{d}$. If $ A \in \R^{d \times d} $ is symmetric, we write $ \lambda_1(A) \geq \lambda_2(A) \geq \dots \geq \lambda_{d}(A) $ to denote {its} eigenvalues in descending order. We also denote $ \lambda_{\max}(A) = \lambda_{1}(A) $ and $ \lambda_{\min}(A) = \lambda_{d}(A) $ --- its largest and smallest eigenvalues, respectively. In particular, we have that $ \| A \| = \max\{ | \lambda_{\max}(A)|, |\lambda_{\min}(A)|\} $. We say that a symmetric matrix $ A $ is positive semi-definite (PSD) if $ v^{\T} A v \geq 0 $ for all $v$. We also write $ A \preceq B $ and $ B \succeq A $ if $ B - A  $ is a PSD {matrix}. Furthermore, $I$ denotes the identity matrix whose dimensions are clear from the context. We denote $ \mathbf{1} = (1, \dots, 1)^{\T} $ of dimension $N$, so that $ w^{\T} \mathbf{1} = \sum_{i = 1}^{N} w_i $.

For a PSD matrix $A \in \R^{d \times d}$ we {denote its \emph{effective rank} by
\[
\mathbf{r}(A) = \operatorname{Tr}(A) / \lambda_{\max}(A) = \sum_{j = 1}^{d} \lambda_{j}(A) / \lambda_{\max}(A).
\]
The effective rank is clearly always smaller than both the dimension $d$ and the matrix rank of $A$. 
This quantity plays an important role in covariance estimation problems. In particular, it was shown by \cite{koltchinskii2017concentration} that in the Gaussian case, the performance of the sample covariance matrix is governed by the effective rank of the covariance matrix and is not sensitive to a potentially larger dimension of the ambient space.}

\section{Mean-Variance and Global Minimum Variance portfolios}
\label{sec:mvgmv}

Suppose we have an opportunity to invest into $N$ assets and $ r_{1}, \dots, r_{N} $ denote their log-returns. Let $ X = (r_1, \dots, r_{N})^{\T} $ be the multivariate return vector with mean $ \mu $ and covariance $\Sigma$. Then a portfolio with allocation weights $w = (w_1, \dots, w_N)^{\T}$ has returns with expectation $ \mu^{\T} w $ and  variance $ w^{\T} \Sigma w $.

One of the fundamental portfolio policies, the \emph{mean-variance}  portfolio (MV), is based on maximizing the utility
\begin{equation*}
    M_{\gamma}(w; \mu, \Sigma) = \mu^{\T} w - \frac{\gamma}{2} w^{\T} \Sigma w \qquad \text{subject to} \;\;  w^{\T} \onev  = 1 \, ,
\end{equation*}
which takes as input the mean $ \mu $ and the covariance operator $\Sigma$. Moreover, $ \gamma $ is a fixed \emph{risk aversion}  parameter provided by the investor.
The quadratic term in the above expression represents the variance of the portfolio return $ \operatorname{Var}(w^{\T} X) $, and the linear term is its mean ${\E} (w^{\T} X)$.

{Some researchers} often discard the dependence on mean and concentrate on an alternative portfolio policy that minimizes the risk measure
\begin{equation}\label{GMV_opt}
    R(w; \Sigma) = \frac{1}{2} w^{\T} \Sigma w \qquad \text{subject to }  w^{\T} \onev  = 1  ,
\end{equation}
which corresponds to finding a \emph{global minimum variance} portfolio (GMV).  The quantity $ w^{\T} \Sigma w = \operatorname{Var}(w^{\T} X) $ is often regarded as \emph{risk} of a portfolio allocation $w$ in the financial literature.

Suppose we have an i.i.d. sample $ X_1, \dots, X_T $ that comes from a distribution with mean $\mu$ and covariance $\Sigma$. Our goal is to construct portfolio allocation weights $\hat{w}$ that are as close to optimum as possible. Below we provide theoretical high-probability guarantees in terms of the gap between the estimator and the population optimal solution, that is, 
\[
    R(\hat{w}; \Sigma) - \min_{w^{\T}\onev} R({w}; \Sigma)
\]
in the GMV case, or
\[
    \max_{w^{\T}\onev} M_{\gamma}({w}; \mu, \Sigma) - M_{\gamma}(\hat{w}; \mu, \Sigma)
\]
for the mean-variance portfolio. Notice that both of these entities are non-negative.

We analyze two different situations, focusing on high-dimensional non-asymptotic bounds. 
Firstly, we consider the hypothetical situation where the covariance matrix $\Sigma$ is well-conditioned, that is when the condition number $\kappa \myeq \lambda_{\max}(\Sigma) / \lambda_{\min}(\Sigma)$ is constant. Such situation is unlikely in practice, and we present the following result mainly out of theoretical curiosity. 
{
\begin{theorem*}[A simplified statement]
Suppose that $ \Sigma $ is well-conditioned, i.e., its condition number is bounded by an absolute constant. Fix $\delta \in (0, 1)$. There is a computationally efficient estimator $\hat{w}_{\delta}$ that satisfies, with probability at least $1 - \delta$,
\[  
    \max_{w^{\T}\mathbf{1} = 1} M_{\gamma}(w; \mu, \Sigma) -  M_{\gamma}(\hat{w}_{\delta}; \mu, \Sigma) \lesssim \|\Sigma\|\frac{N \log N + \log(1/\delta)}{T} \,, 
\]
even when the distribution is non-Gaussian and has heavy tails.
\end{theorem*}
}
The above result demonstrates that the MV portfolio can be robustly estimated as long as the ratio $ N \log N / T$ remains small. It is known that for convergence to the optimal risk one has to have that $ N / T = o(1) $. For instance, \cite{karoui2013realized} considers a ``large $N$, large $T$'' situation where $N/T$ converges to some constant $ \gamma \in (0, 1)$, and shows that there is a constant gap between the realized risks of the empirical and the optimal solutions.  More recently, \cite{bartl2021monte} studied a similar portfolio optimization setup in a well-conditioned case. Although their algorithm is robust with respect to heavy-tailed data and achieves similar rates of convergence, their estimator is not computationally feasible.

Secondly, we consider the case where $\Sigma$ is allowed to be ill-conditioned. {In this case, the condition number $\kappa$ can be large and the effective rank $\mathbf{r}(\Sigma) = \operatorname{Tr}(\Sigma)/\|\Sigma\|$ can be much smaller than $N$.} This case corresponds to a less regular optimization problem and we provide slower convergence guarantees with respect to the number of observations. 
{
\begin{theorem*}[A simplified statement]
There is a computationally efficient estimator $\hat{w}_{\delta}$, such that, with probability at least $1 - \delta$,
\[  
     R(\hat{w}_{\delta}; \Sigma) - \min_{w^{\T}\mathbf{1} = 1} R(w; \Sigma) \lesssim \|\Sigma\|\sqrt{\frac{\mathbf{r}(\Sigma) \log \mathbf{r}(\Sigma) + \log(1/\delta)}{T}} \, ,
\]
even when the distribution is non-Gaussian and has heavy tails.
\end{theorem*}
}
This result suggests that the GMV portfolio converges to optimum as long as $ \mathbf{r}(\Sigma) $ is much smaller than $ T $, which is a rather adequate assumption.
For example, for the S\&P100 dataset, we evaluate that $ \mathbf{r}(\Sigma) \approx 3 $ and for the Russell3000 constituents, $ \mathbf{r}(\Sigma) \approx 7 $. Moreover, implying that effective rank is also a measure of effective dimensionality of the investment universe, this measure also reflects the conclusions of widely used three and five Fama-French factor models, who proved that a cross-sectional variation in average stock returns can be explained with only three   \cite{fama1992cross}  extended later to five observable factors, \cite{fama2015five}. 
Notice also that the condition number $\kappa = \lambda_{\max}(\Sigma) / \lambda_{\min}(\Sigma)$ is bounded from below by $ N / \mathbf{r}(\Sigma) $, hence the covariance matrix is indeed ill-conditioned in these two applications.

\subsection{Recent advances in robust statistics}
\label{sec:advances}
The covariance matrix and the mean
are not known in practice and must be estimated
based on the observed log-returns. 

In an idealized situation where $ X_i \sim \mathcal{N}(\mu, \Sigma) $ are Gaussian, we have that the standard empirical mean estimator $ \hat{\mu} = T^{-1} \sum_{i = 1}^{T} X_i $ provides one with optimal high probability deviation bounds. In particular, for {any} $ \delta \in (0, 1) $, we have that, with probability at least $ 1-  \delta$,
\begin{equation}
    \label{gaussian_emp_mean}
    \| \hat{\mu} - \mu \| \le \sqrt{\frac{\Tr(\Sigma)}{T}}  + \sqrt{\frac{2\| \Sigma \| \log(1/\delta)}{T}}
\end{equation}
{See Example 5.7 in \citet{boucheron2013concentration} for this derivation. The sharp deviation term $\sqrt{\frac{2\| \Sigma \| \log(1/\delta)}{T}}$ is very specific to the Gaussian assumption and could not be expected for less regular distributions. In particular, } here the dependence on the confidence level is logarithmic and additive, in the sense that the bound separates into the \emph{strong term}  scaled with $\sqrt{\Tr(\Sigma)}$ and corresponding to the error on average, and the \emph{weak term} that is scaled with~$\sqrt{\|\Sigma\|}$. The weak term can potentially be a lot smaller than the strong one, even for very small values of $\delta$.

Similarly, \cite{koltchinskii2017concentration} proved that in the case of i.i.d. zero mean Gaussian observations, the sample covariance $ \hat{\Sigma} = T^{-1} \sum_{i = 1}^{T} X_iX_i^{\T}$ satisfies the following deviation bound. For any $\delta \in (0, 1)$, with probability at least $1-\delta$,
\begin{equation}
\label{eq:koltchlounici}
    \| \hat{\Sigma} - \Sigma \| \lesssim \|\Sigma\|\sqrt{\frac{\mathbf{r}(\Sigma) +\log(1/\delta)}{T}},
\end{equation}
whenever $T \gtrsim \mathbf{r}(\Sigma) + \log(1/\delta)$. For the version of this inequality with explicit constants we refer to \citet{zhivotovskiy2021dimension}.

Having these performance bounds in mind, one is interested if the same bounds can be achieved under milder assumptions in a computationally efficient manner.
\cite{lugosi2019sub} developed an estimator that matches the bound \eqref{gaussian_emp_mean} {up to multiplicative constant factors} under the only assumption that the covariance matrix exists (i.e., {the} two moments assumption). Loosely speaking, they propose to control the deviations of the median-of-means of the projections $ X_i^{\T} v $ uniformly in all directions $ v \in \R^N $. Based on this bound, {they came up with an estimator, which is however not practical}. Further developments were made in \citet{lugosi2021robust}; see also \citet{lugosi2019mean} for a thorough review on this topic. 

\cite{hopkins2018sub} first proposed an estimator that can be computed in polynomial time, and the time complexity was subsequently improved to nearly linear  \citet{cherapanamjeri2019fast, depersin2019robust, hopkins2020robust}. An alternative method called \emph{spectral sample reweighing} was developed in the context of robust estimation with outliers. Given data points $ \{ x_i\}_{i = 1, \dots, k} $ the goal is to reweigh the points $x_i$ with some weights $ u_{i} \in [0, 1]$ and find a center $ \nu \in \R^{N} $ such that the largest eigenvalue of the weighted covariance $ \sum_{i} u_i (x_i - \nu)(x_i - \nu)^{\T} $ is small. \citet*{hopkins2020robust} develop an algorithm that does this in nearly linear time; see also \citet{diakonikolas2017being, zhu2020robust}. More importantly for us, \citet*{hopkins2020robust} establish a direct connection between the sample reweighing and the method developed in \citet{lugosi2019sub}, which makes this approach applicable in the heavy-tailed {setup} as well. We discuss this connection in greater detail in Section~\ref{sec:proof_a_hat}.

{The problem of robust covariance estimation is more challenging. \cite{mendelson2020robust} construct an abstract estimator that matches the bound \eqref{eq:koltchlounici} up to some logarithmic factors and a more recent paper \citet{abdalla2022covariance} removes the remaining logarithmic factor in the bound. Unfortunately, it is still not known whether such performance can be achieved with a computationally efficient algorithm. Existing computationally efficient implementations are showing suboptimal statistical guarantees \citet{ke2019user, ostrovskii2019affine, cherapanamjeri2020algorithms} and sometimes require additional assumptions such as the so-called SoS-hypercontractivity that are hard to verify in the non-Gaussian situation.
Moreover, \citet{cherapanamjeri2020algorithms} conjecture that as long the median-of-means approach is used, it is algorithmically hard to robustly estimate the sample covariance matrix in the presence of heavy-tailed data.}

After this short excursion to some recent results in robust estimation, let us now come back to our portfolio optimization problem. Since we cannot get our hands on a robust covariance estimator, we take another route by observing that both MV and GMV are convex optimization problems. 
\subsection{Gradient descent for portfolio optimization}
\label{sec:pgd}
Our goal is to avoid the estimation of the whole covariance matrix, but rather resort to the estimation of the action of this operator $ \Sigma w $ on some limited set of vectors $w$. We will use a procedure based on \emph{projected gradient descent} (PGD), which is a standard convex optimization method.
For instance, if we want to minimize the GMV objective with known $ \Sigma$, the following sequence of approximations converges to an optimal solution (which is not necessarily unique): we start with arbitrary initial vector $w_0$ and then take the update steps,
\begin{equation}\label{GD_for_GMV}
    w_{s} = \Pi_{1} [w_{s-1} - \eta \nabla_{w} R(w_{s-1}; \Sigma)],
    \qquad
    s = 1, 2, \dots,
\end{equation}
where $\Pi_{1}$ is the orthogonal projector onto the restricted (convex) set $ \{w: \; w^{\T}\mathbf{1} = 1\} $, which can be explicitly defined by the mapping,
\begin{equation}
    \label{Pi_1_definition}
    \Pi_{1} x = (I - N^{-1} \onev\onev^{\T}) x + N^{-1} \onev \, .
\end{equation}
It is straightforward to see that $ R(\,\cdot\,; \Sigma)$ is convex (since the covariance operator is positive semi-definite), and $\| \Sigma\|$-smooth {(we say that a continuously differentiable function $f$ is $\beta$-smooth if the
gradient $\nabla f$ is $\beta$-Lipschitz, that is,
$\|\nabla f(x) - \nabla f(y)\| \le \beta \|x - y\|$)}. By Theorem~{3.7} from \citet{bubeck2014convex} the sequence \eqref{GD_for_GMV} converges to a minimum at a rate $1/s$ as long as $ \eta \leq 1/ \| \Sigma \|$. Moreover, in the case where $ \Sigma$ is non-degenerate, the objective becomes strongly convex, and the sequence converges at a faster exponential rate under the same requirement on the step size, see Theorem 3.10 in \citet{bubeck2014convex}.

The case of MV portfolio is similar, only this time we need to maximize a concave function instead of minimizing a convex one. If we replace $\nabla_{w} R(w_{s-1}; \Sigma)$ with $ -\nabla_{w} M_{\gamma}(w_{s-1}; \Sigma)$ in \eqref{GD_for_GMV}, then by the same reasons, the sequence
converges to the maximum of $M_{\gamma}$ as long as $ \eta \leq (\gamma \| \Sigma \|)^{-1}$, with exponential rate when $\Sigma$ is non-degenerate.

The PGD iterations require computation of the following gradients,
\[
    \nabla_{w}R(w; \Sigma) = \Sigma w
    \qquad \text{or} \qquad
    \nabla_{w}M_{\gamma}(w; \Sigma) = \mu - \gamma \Sigma w.
\]
where the mean $\mu$ and covariance $\Sigma$ are typically replaced with their empirical counterparts that are calculated using given historical observations $X_1, \dots, X_T$. As discussed in the previous section, there is a practical robust mean estimator in \citet{hopkins2020robust} with all desired properties. Since such {an} estimator is not available for the covariance operator, we instead produce an estimator $\hat{a}_{\delta}(w)$ for the PGD increment that estimates $ \Sigma w$ for each $w$ separately, and plug it into the update steps of PGD. 

To see how it can be done, suppose for a moment that the expectation vanishes. Then, we can represent this product as a mean of a random vector as follows,
\[
    \Sigma w = {\E XX^{\T} w\, }.
\]
We therefore can apply the robust mean algorithm to the vectors {$X_iX_i^{\T}w$} and obtain a robust estimator of $\Sigma w$. However, we need to take additional care to ensure that the estimator is an appropriate approximation uniformly in all directions. For this, we slightly adjust the procedure in the spirit of \citet{mendelson2020robust}. In the latter work, the only assumption used is the equivalence of the fourth and the second moments in all directions {sometimes called the \emph{bounded kurtosis} assumption.}

\begin{assumption}[Bounded kurtosis]
\label{bounded_kurtosis_assumptino}
The return vectors $ X_1, \dots, X_T $ are i.i.d. observations of a random vector $X  $, that has mean $\mu$, covariance $\Sigma$, and satisfies for all $ u \in \R^{N} $,
\begin{equation}\label{bounded_kurtosis_eq}
    \E^{1/4} |u^{\T}(X - \mu)|^{4} \leq K \E^{1/2} |u^{\T}(X - \mu)|^{2},
\end{equation}
where $ K \geq 1$ is some fixed constant. For the rest of the paper, we ignore the explicit dependence on $K$ in our bounds treating $K$ as an absolute constant.
\end{assumption}

All our results will be stated under this assumption. Importantly, no other assumptions are made on the covariance matrix $\Sigma$ throughout the paper. This makes our setup quite different from traditional factor models where one assumes some specific structure of $\Sigma$. The assumption is clearly satisfied for Gaussian random vectors with $K$ being an absolute constant. However, we even cover some heavy-tailed distributions. {For some heavy-tailed examples where Assumption \ref{bounded_kurtosis_assumptino} is satisfied we refer to \citet{mendelson2020robust}}. In the context of covariance estimation, a robust estimator is one that has Gaussian deviation bounds (i.e., as in \eqref{eq:koltchlounici}) but only requires the underlying distribution to follow the bounded kurtosis assumption. The next step is to provide a robust estimator of $ \Sigma w $ that works simultaneously in all directions. Recall that $\mathbf{r}(\Sigma)$ denotes the effective rank 
$
\mathbf{r}(\Sigma) = \operatorname{Tr}(\Sigma) / \| \Sigma\| \, .
$

\begin{proposition}\label{a_hat_prop}
Suppose that Assumption \eqref{bounded_kurtosis_assumptino} holds. Fix $\delta \in (0, 1)$
There is a computationally efficient estimator $ \hat{a}_{\delta}(w) $ that depends on direction $w$ and $T$ i.i.d. observations and such that, with probability at least $1-\delta$,
\[
    \| \hat{a}_{\delta}(w) - \Sigma w\| \lesssim  \| \Sigma \| \sqrt{\frac{\mathbf{r}(\Sigma) \log \mathbf{r}(\Sigma) + \log(1/\delta)}{T}}  \| w\|,
\]
uniformly for all vectors $w$. 
\end{proposition} 
We postpone the proof and detailed description of the estimation algorithm until Section~\ref{sec:proof_a_hat}. {We note that its computational complexity is $\widetilde{O}(TN + T^{1/2} N \log(1/\delta)^{2})$, where the notation $ \widetilde{O} $ suppresses some multiplicative logarithmic factors.}

\begin{remark}
We expect that the proposed estimator is robust to finite number of outliers as well, as typically happens with estimators based on \cite{lugosi2019mean}. To avoid technically involved proofs, we do not check this rigorously, instead we demonstrate robustness to outliers in the simulation study in Section~\ref{sec:sim}.
\end{remark}

\begin{remark}
For technical reasons, the estimator $ \hat{a}_{\delta}(w) $ depends on a {norm-truncation} parameter $ D $ that needs to be of order~$ \left(\frac{T \mathbf{r}(\Sigma)}{\log \mathbf{r}(\Sigma)}\right)^{1/4} $, which is unknown in general.
{It appears that since $D$ increases with $T$, in many natural situations this truncation parameter is of a much larger order than {$\max\limits_i\|X_i\|$} and can be mostly ignored in practice. For more details see Section \ref{sec:proof_a_hat}.} 
\end{remark}

\begin{remark}
We point out that when one has access to some covariance estimator $\hat{\Sigma}$, one can simply take a family of estimators $ \hat{a}(w) = \hat{\Sigma} w $. For instance, in the Gaussian case, taking the standard empirical covariance estimator would yield thanks to \eqref{eq:koltchlounici},
\[
    \| \hat{a}(w) - \Sigma w\| \lesssim  \| \Sigma \| \sqrt{\frac{\mathbf{r}(\Sigma)  + \log(1/\delta)}{T}}  \| w\| \ ,
\]
with probability at least $1 - \delta$, uniformly for all $w$. The estimator of Proposition \ref{a_hat_prop} achieves the same rate of convergence under minimal distributional assumptions.
\end{remark}

Now we can plug in the estimator of $\Sigma w$ (appearing in Proposition \ref{a_hat_prop}) into the update rule \eqref{GD_for_GMV}. To be precise, in the case of GMV optimization, our updates look as follows,
\[
    w_{s} = \Pi_{1} \left[ w_{s-1} - \eta \hat{a}_{\delta}(w_{s-1}) \right], \qquad s = 1, 2, \dots
\]
Naturally, the error may accumulate with each update, and we need to carefully analyze  how the resulting solution differs from the optimum, to which the sequence \eqref{GD_for_GMV} converges. We analyze this update rule in two separate cases.
 
First, we consider the case of a well-conditioned matrix $ \Sigma $, meaning that the ratio of its maximal and minimal eigenvalues is bounded by a constant. 
This means that the problem of maximizing the MV utility is a strongly-convex optimization problem, so that the gradient descent sequence enjoys exponential convergence rate and, as we show below, the error of estimation does not accumulate. However, in that situation the effective rank $\mathbf{r}(\Sigma) = \Tr(\Sigma) / \| \Sigma \|$ is of order $N$, so the convergence only works in the case where  $ N /T  $ is small. Moreover, in typical applications, the covariance matrix is ill-conditioned, which is one of the reasons the MV portfolio performs so poorly. 
For instance, this can be checked through evaluation of the effective rank: for the S\&P100 dataset we estimate $ \mathbf{r}(\Sigma) \approx 3 $ and for the Russell3000 set we estimate that $ \mathbf{r}(\Sigma) \approx 7 $, in both cases much smaller than the dimension $N$. This brings us to the second part of our GD analysis, where we only consider the case of GMV optimization with ill-conditioned covariance matrix that has small effective rank. This scenario corresponds to non-strongly convex optimization and has weaker convergence rate. However, it enjoys dimension-free bounds, meaning that the convergence is guaranteed as long as the number of observations is much larger than $\mathbf{r}(\Sigma)$, regardless of how high the total number of assets is. We also point out that in this case, one has to stop after an appropriate number of steps to avoid overfitting.

\section{Well-conditioned case}
\label{sec:well_cond}

For maximizing the MV utility $ M_{\gamma}(w; \mu, \Sigma) $, we consider the following updates,
\begin{equation}\label{well_cond_updates}
    w_{s} = \Pi_{1}\left[ w_{s-1} + \eta(\hat{\mu} - \gamma \hat{a}(w_{s-1})) \right],
    \qquad
    s = 1, 2, \dots
\end{equation}
where $ \hat{\mu} $ is some estimator of mean $ \mu $, and $ \hat{a}(w) $ is some family of estimators for the action of covariance operator $ \Sigma w $. We first show a deterministic result that controls the convergence through the errors of estimators $ \hat{\mu} $ and $ \hat{a}(w) $. 

\begin{lemma}\label{well_cond_convergence}
Denote, $ {w}^* = \arg\max_{w^{\T}\mathbf{1} = 1} M_{\gamma}(w; \mu, {\Sigma})$. Suppose that we have {an} access to an estimator $\hat{\mu}$ satisfying 
$$
    \| \hat{\mu} - \mu\| \leq \Delta_{\mu},
$$ 
and {an} access to a family of estimators $\hat{a}(w)$ satisfying uniformly for all $ w\in \R^{N} $,
$$ 
    \| \hat{a}(w) - \Sigma w \| \leq \Delta_{\Sigma} \| w \| \, .
$$
Let $ \lambda_{\max}$, $\lambda_{\min}$ denote the maximal and minimal eigenvalues of $\Sigma$, respectively. Assume that $ \eta \leq 1/(\gamma \lambda_{\max}) $ and $ \Delta_{\Sigma} < \lambda_{\min} $.  Then, the sequence \eqref{well_cond_updates} satisfies
\[
    \| w_{s} - w^{*} \| < \left(1 - {\gamma\eta(\lambda_{\min} -\Delta_{\Sigma})}\right)^{s} \| w_{0} - w^{*} \| + \frac{ {\gamma^{-1}} \Delta_{\mu} + \Delta_{\Sigma}\| w^{*}\|}{\lambda_{\min} - \Delta_\Sigma} \, .
\]
\end{lemma}

We now apply this lemma to the case where we use $ \hat{\mu}_{\delta} $ from \citet{hopkins2020robust} and $ \hat{a}_{\delta}(w) $ from Proposition~\ref{a_hat_prop}. {Theorem D.3 from \citet{hopkins2020robust} shows that their estimator satisfies the sub-Gaussian bound \eqref{gaussian_emp_mean}. That is,} with probability at least $1-\delta$,
\begin{equation}\label{mu_bound}
    \| \hat{\mu}_{\delta} - \mu \| \lesssim \| \Sigma \|^{1/2} \sqrt{\frac{\mathbf{r}(\Sigma) + \log(1/\delta)}{T}}
\end{equation}
Furthermore, by Proposition~\ref{a_hat_prop}, with probability at least $1-\delta$, simultaneously for all $w \in \R^N$,
\begin{equation}\label{a_bound}
    \| \hat{a}_{\delta}(w) - \Sigma w\| \lesssim \| \Sigma \| \sqrt{\frac{\mathbf{r}(\Sigma) \log \mathbf{r}(\Sigma) + \log(1/\delta)}{T}} \| w \| \, .
\end{equation}
{Substituting the two error terms \eqref{mu_bound}, \eqref{a_bound} into Lemma~\ref{well_cond_convergence}, we arrive at the following result. We postpone the derivations to Section~\ref{proof_well_cond_main_res}.}

\begin{theorem}\label{well_cond_main_res} Suppose, we are given independent $ X_{1}, \dots, X_T $ that have mean $\mu$ and covariance $\Sigma$, and the distribution satisfies the bounded kurtosis assumption \eqref{bounded_kurtosis_eq}.
Let $ \kappa = \lambda_{\max}(\Sigma) / \lambda_{\min}(\Sigma) $ {denote} the condition number. There is an absolute constant $ C>0$, such that the following holds. {If for $\delta \in (0, 1)$, the following holds}  
\begin{equation}\label{T_lower_bound}
    T \geq C \kappa^{2} \left(\mathbf{r}(\Sigma)\log\mathbf{r}(\Sigma) + \log(1/\delta)\right),
\end{equation}
{then} there is an estimator $ \hat{w}_{\delta} $ such that, with probability at least $1-\delta$,
\[  
    \max_{w^{\T}\mathbf{1} = 1} M_{\gamma}(w; \Sigma, \mu) -  M_{\gamma}(\hat{w}_{\delta}; \Sigma, \mu) \lesssim \kappa^{2} \left({\gamma^{-1} + \gamma} \| \Sigma\| \|w^{*}\|^{2} \right)   \frac{\mathbf{r}(\Sigma)\log\mathbf{r}(\Sigma) + \log(1/\delta)}{T} \, .
\]
\end{theorem}

The above result has a number of favorable  properties:
\begin{itemize}
    \item The estimator only requires $\mathcal{O}(\log T)$ gradient descent updates. In addition, the amount of steps only has to be sufficiently large, i.e., there is no danger of overfitting by running the gradient descent for too long;
    \item the bound scales with $1/T $ when all the other parameters are fixed. In the optimization {and statistics} literature, this is regarded as a fast rate convergence. This rate is typical for strongly convex stochastic optimization problems;
    \item The value {$\| w \|^{2}$} is often considered as a diversification measure of an allocation strategy, see \citet{strongin2000beating}. For instance, for the EW portfolio {this} value is $1/N$. One may expect it to be very small for the optimal portfolio.
\end{itemize}
However, the dependence on the condition number of the covariance matrix outweighs some of these useful properties. It is straightforward to verify that $ \kappa \mathbf{r}(\Sigma) \geq N $.
Hence, the above result only works in the setting, where the number of observations $T$ is greater than the dimension $N$. The remaining term $ \kappa $ may further worsen the bound, so our result is rather limited to well conditioned covariance matrices. Unfortunately, it is rarely the case in practice: for our two datasets we estimate that $ \mathbf{r}(\Sigma) \approx 3 $ for S\&P100 and $ \mathbf{r}(\Sigma) \approx  7 $ for Russell3000. The naive lower bound $\kappa \geq N/\mathbf{r}(\Sigma)$ yields that $ \kappa \geq 27 $ for S\&P100 and $ \kappa > 250 $ for Russell3000. Therefore, our result does not contradict a commonly accepted evidence that MV portfolios perform poorly even when $T$ is moderately larger than $N$  \citet{ao2019approaching}.

\section{Ill-conditioned case}
\label{sec:ill_cond}

We now consider the case where we have no control over the condition number of $\Sigma$ and it can even be degenerate. We will state the bound in the regime where only the effective rank $ \mathbf{r}(\Sigma) $ has to be small, and no requirements on the total dimension $N$ are needed. In this section, we only consider the GMV portfolio.

For minimizing the GMV risk $ R(w; \Sigma) $, we consider the following updates,
\begin{equation*}
    w_{s} = \Pi_{1}\left[ w_{s-1} - \eta \hat{a}(w_{s-1})\right],
    \qquad
    s = 1, \dots, S
\end{equation*}
where  $ \hat{a}(w) $ is some family of estimators for the action of covariance operator $ \Sigma w $.

Similarly to the previous section, we first show a deterministic result that controls the convergence through the error of this estimator.

\begin{lemma}\label{lemma_ill_cond}
Denote, $ w^{*} = \arg\min_{w^{\T}\mathbf{1} = 1} R(w; \Sigma)$. Suppose that we have {an} access to a family of estimators $\hat{a}(w)$ satisfying uniformly for all $ w\in \R^{d} $,
$$ 
    \| \hat{a}(w) - \Sigma w \|_{2} \leq \Delta_{\Sigma} \| w \|_2 \, .
$$
Assume that $ \eta \leq 1/\lambda_{\max} $ and let the number of {weights updates $S$ satisfy $ S \eta \Delta_{\Sigma} \leq 1$}.  Then,
\[
    R(w_{S}; \Sigma) - R(w^{*}; \Sigma) \lesssim \max\{ \| w_{0} - w^{*}\|, \| w^{*}\|\}^{2} \left( \frac{1}{\eta S} + \eta \Delta_{\Sigma}^{2} S \right) .
\]
{In particular, }for the optimal choice $ S \sim 1 / (\eta \Delta_{\Sigma}) $, we have
\[
    R(w_{s}; \Sigma) - R(w^{*}; \Sigma) \lesssim \max\{ \| w_{0} - w^{*}\|, \| w^{*}\|\}^{2} \Delta_{\Sigma} .
\]

When the solution to $\min_{w^{\T}\mathbf{1} = 1} R(w; \Sigma)$ is not unique, it is sufficient to pick any $ w^{*} \in \mathrm{Argmin}_{w^{\T} 1 = 1} R(w; \Sigma) $ and the result still holds.
\end{lemma}

Once again, we plug our estimator $ \hat{a}_{\delta}(w)  $ into the update rule. In addition, we take the initial approximation to be {an} EW portfolio. Namely, our sequence {is} as follows
\begin{equation}\label{a_hat_delta_updates}
    w_{0} = N^{-1} \mathbf{1}, \qquad w_{s} = \Pi_{1} [w_{s-1} - \hat{a}_{\delta}(w_{s-1})],
    \qquad s= 1, \dots, S \, .
\end{equation}

\begin{corollary}
Suppose that Assumption~\ref{bounded_kurtosis_assumptino} holds. Take $ \eta = 1/\lambda_{\max} $ and $ S \sim T (\mathbf{r}(\Sigma) \log \mathbf{r}(\Sigma) + \log(1/\delta))$, and set $ \hat{w}_{\delta} = w_{S} $. Then, with probability at least $1- \delta$,
\[
    R(\hat{w}_{\delta}; \Sigma) - R(w^{*}; \Sigma) \lesssim \| \Sigma \| \| w^{*}\| ^{2}  \sqrt{\frac{\mathbf{r}(\Sigma) \log \mathbf{r}(\Sigma) + \log(1/\delta)}{T}} \, .
\]
\end{corollary}
\begin{proof}
Simply substitute the bound \eqref{a_bound} into Lemma~\ref{lemma_ill_cond}. We also notice that $ \| w^{*} - w_0\|^{2} \leq 2\| w^{*}\|^{2} + 2\|w_0\|^{2} $, and that $ \| w^{*}\|^{2} \geq (\mathbf{1}^{\T} w^{*})^{2} / N  = \| w_{0}\|^{2} $.
\end{proof}

\begin{remark}
We remark that the scaling value $ \| \Sigma \| \| w^{*}\|^{2} $ is only an upper bound on the optimal risk $ R(w^{*}; \Sigma) = \frac{1}{2} (w^{*})^{\T} \Sigma w^{*}$ and we cannot guarantee a ratio-type bound of the form $ R(\hat{w}_{\delta};\Sigma) = (1 + o(1))R(w^{*}; \Sigma)  $. However, this is not uncommon. For instance, \cite{fan2012vast} {show} that a portfolio with GEC {(Gross Exposure Constraints)} constraints $ \sum_{i} |w_i| \leq C $ satisfies,
\[
    R(\hat{w}; \Sigma) - R(w^{*}; \Sigma) \leq (1 + C)^{2} \max_{ij} | \hat{\Sigma}_{ij} - \Sigma_{ij}|,
\]
where one typically has a bound $\max_{ij} | \hat{\Sigma}_{ij} - \Sigma_{ij}| \lesssim \sqrt{(\log N) / T}$. Our bound {is more beneficial when the optimal portfolio} is well-diversified (i.e., $\| w^{*}\|^{2} \sim 1/N$), even though we do not impose any restrictions on the selected portfolio. 
\end{remark}

\section{Evaluation of empirical results}
\label{sec:bench}
To test the performance of our approach, we apply it to two data sets of stocks. The first data set consists of 81 constituents of S\&P100 index (as on January 1, 2021) and covers time span from January 2, 2000 to December 31, 2020 summing up to 5282 daily log-returns. These 81 stocks  have a continuous return time-series over the period of our study.  The second data set consists of 600 random constituents of Russell3000 index as on January 1, 2021, period of analysis is limited by 11 years: from January 2, 2010 to December 31, 2020. The length of analyzed time series is 2768 observations. 

For the portfolio construction, we employ a rolling-window approach with monthly rebalancing. Specifically, we choose an estimation
window of length $T$ days starting on date $T + 1$, for each rebalancing period $l$ (${l = 1, \ldots, L}$, with $L$ the number of rebalancing periods) we use the data in the previous $T$ days to estimate
the parameters required to implement a particular strategy.  The input parameters are estimated using daily returns of the most recent 12 and 24 months, corresponding roughly to 252 (500) daily returns of past data (with the length of estimation windows   $T=252$ and $T = 500$). Results for the window length 24 months are discussed in  \refsec{sec:emp} and summary tables for the length window 12 months are provided in \refsec{sec:perf}. Thus, the out-of-sample period for the S\&P100 data set starts on January 2, 2002 with 4781 observations, which corresponds to the number of rebalancing periods $L = 216$, and for the Russell3000 data set --- on January 3, 2012 with 2265 out-of-sample observations corresponding to $L = 108$. The source for both data sets is Thompson Reuters.

 \subsection{Benchmark portfolios}
Here we present the empirical results for GMV portfolio and evaluate its relative performance.
The allocation rules included into the empirical study with corresponding reference and abbreviations are listed in \reftab{Benchmark_portfolios}.

\paragraph{Equally weighted (EW).}  \cite{demiguel2009portfolio} 
argue that a naive allocation strategy with weights $w_i = 1/N$ is hard to outperform in practice. It is often used as a benchmark for comparative analysis. 

\paragraph{Sample-based Global minimum portfolio (GMV).} This is the most straightforward way to GMV optimization. The sample covariance matrix $\hat{\Sigma}$ is plugged into the objective in \eqref{GMV_opt}. We should note that this strategy is only included for S\&P100 data set, since for the Russell3000 we have $N > T$, and the sample covariance matrix is not invertible.   

\paragraph{Global minimum portfolio with short-sale constraint (GMV\_long).} This portfolio is a sample-based GMV portfolio with only long positions allowed. This means that GMV objective corresponding to the empirical covariance is optimized subject to the constraints $ w_j \geq 0 $.

\paragraph{Global minimum portfolio  with linear shrinkage estimator (GMV\_lin).} \cite{ledoit2004well}  propose an asymptotically optimal convex linear combination of the sample covariance matrix $\hat{\Sigma}$ with the identity matrix. Optimality is meant with respect to a quadratic loss function, asymptotically, as the number of observations and the number of assets go to infinity together.  \cite{ledoit2004well} use as a covariance matrix estimator as a convex linear combination of the sample covariance matrix and the identity matrix (shrinkage target) as follows:
\begin{equation*}
    \hat{\Sigma}_{shrink} =\rho I+(1 - \rho) \hat{\Sigma},
\end{equation*}
where $\rho$ is the shrinkage intensity parameter and $\hat{\Sigma}$ is the sample covariance matrix. \href{https://www.econ.uzh.ch/en/people/faculty/wolf/publications.html}{Their R code} is used in this horse race exercise.
\paragraph{Global minimum portfolio with non-linear shrinkage estimator (GMV\_nlin).} 

\cite{ledoit2017nonlinear} use the spectral decomposition for the empirical covariance 
\begin{equation*}
    \hat{\Sigma} \myeq U \widehat{D} U^{\T} 
\end{equation*}
where $\widehat{D} \myeq \operatorname{diag}\left(\widehat{d} \left(\lambda_{ 1}\right), \ldots, \widehat{d}\left(\lambda_{N}\right)\right)$, where $\lambda_{1},  \dots, \lambda_{N}$ are the sample eigenvalues, and $\hat{d}$ is some nonlinear cutoff threshold based on value of $N/T$ and the magnitude of the eigenvalues $\lambda_j$.

\begin{table}
\centering
\scalebox{1}{\begin{tabular}{lll}
 \toprule
Model & Reference & Abbreviation\\
 \midrule
Equally weighted	&	\cite{demiguel2009optimal} 	&	\textcolor{EW}{EW}	\\
Robust Global Minimum Variance &		&	\textcolor{GMV_{robust}}{GMV\_{robust}}	\\
GMV with sample covariance &	\cite{merton1980estimating} &	\textcolor{GMV}{GMV}		\\
GMV with linear shrinkage cov estimator &	\cite{ledoit2004well} &	\textcolor{GMV_lin}{GMV\_{lin}}		\\
GMV with non-linear shrinkage cov estimator &	\cite{ledoit2017nonlinear} &\textcolor{GMV_nlin}	{GMV\_nlin}		\\
GMV with short sale constraint &	\cite{jagannathan2003risk} & \textcolor{GMV_long}	{GMV\_long}	\\	
Index &	S\&P100 and Russell3000 & \textcolor{Index}	{Index}	\\
 \bottomrule
\end{tabular}}
\caption{Benchmark portfolios}
\label{tab:Benchmark_portfolios}
\end{table}

\subsection{Performance measures}
 We report the following five out-of-sample performance measures for each benchmark portfolio rule. 
 \begin{itemize}
\item \emph{Turnover (TO)}
The main practical objective of the introduced methodology is stabilizing of portfolio weights, aiming at reduction of transaction costs. To assess the impact of potential trading costs associated with portfolio rebalancing, we calculate two measures for turnover. First, following \cite{demiguel2009optimal} and \cite{demiguel2009portfolio},  we present  Turnover, which is defined as an average sum of the absolute value of the rebalancing trades across the $N$ assets of the investment universe and over the $L$ rebalancing months \eqref{TO}. 
\begin{equation}\label{TO}
TO=L^{-1} \sum_{l=1}^{L} \sum_{j=1}^{N}\left|\hat{w}_{ j, l+1}-\hat{w}_{j, l+}\right|\, .
\end{equation}
where $\hat{w}_{j,l}$ and $\hat{w}_{j,l+1}$ are the weights assigned to the asset $j$ for rebalancing periods $l$ and $l + 1$ and $\hat{w}_{j,l+}$ denotes its weight just before rebalancing at ${l + 1}$. Thus, one accounts for the price change over the period, as one needs to execute trades to rebalance the portfolio towards the ${w_l}$ target. High turnover will imply significant transaction costs; consequently, the lower TO of a strategy, the less its performance would be harmed by non-zero transaction costs.
Whereas, for market makers, the reduced turnover is not a crucial parameter in a decision-making process, for buy-side investors, we believe, that a lower turnover could be an important advantage.

\item \emph{Target Turnover (TTO)}

Further, following \citet{petukhina2021investing} we also calculate a target turnover, which is constructed as follows: 

\begin{equation*}
  {TTO} = L^{-1}\sum_{l=1}^{L}{\sum_{j=1}^N}\vert
  \hat{w}_{j,l+1}-\hat{w}_{j,l}\vert \, .
  \label{target-turnover}
  \end{equation*}

In contrast to equation \eqref{TO}  this definition of turnover implies by construction a value of zero for the EW portfolio. We provide this measure to focus on modifications of the target portfolio weights due to active management decisions and cleaned from the influence of assets’ price dynamics. 

We calculate the following four performance measures for both, gross and net, returns series. Gross returns are raw returns that do not take transaction costs into account and net returns with subtracted fees. When the portfolio is rebalanced at time $l+1$, it gives rise to a trade in each asset of magnitude $\left|\hat{\mathrm{w}}_{k, j, l+1}-\hat{\mathrm{w}}_{k, j, l^{+}}\right|$. Thus, to calculate the net return we reduce the return by the cost of such a trade over all assets, given by
$$tc = c\sum_{j=1}^N\left|\hat{\mathrm{w}}_{k, j, l+1}-\hat{\mathrm{w}}_{k, j, l^{+}}\right|,$$
 where $c$ is the proportional transaction cost. In our empirical study we set a proportional transactions cost of 50 basis points per transaction. The same level was accepted in e.g., \cite{demiguel2009optimal}, \cite{fays2011}.
     \item \emph{Cumulative wealth (CW)}  
     
\emph{CW} generated by each benchmark strategy  with initial investment $W_0 =  1 $ is computed as follows for gross returns series:
\begin{equation*}
  W_{l+1} = W_{l} + {\hat{w}_{l}^{\top}X_{l+1}}. 
  \label{cumret}
\end{equation*}
CW for net returns series:
\begin{equation*}
  W_{l+1} = W_{l} + {\hat{w}_{l}^{\top}X_{l+1}} - tc. 
\end{equation*}
 \item \emph{Standard Deviation (SD)}

We compute $SD$ of out-of-sample daily
 returns series.
\item \emph{Sharpe ratio (SR)}
\begin{equation*}
  SR = \frac{{AV}}{{SD}},
  \label{SRatio}
\end{equation*}
where $AV$ and $SD$  are average out-of-sample returns and their standard deviations for each strategy are calculated based on daily out-of-sample returns series.
 We also test, whether GMV\_{robust} delivers a better risk measured by SD, and risk-adjusted performance, measured by SR, than other portfolios at a level that is statistically significant. A two-sided $p$-values for the null hypothesis of equal SD and SR are obtained by  HAC method described in   \citet[Section 3.1]{ledoit2011robust} for SD and in  \citet[Section 3.1]{ledoit2008robust}  for SR.
\item \emph{Calmar Ratio (CR)}

We also compare  \emph{(CR)}, metric for risk-adjusted return, which is widely used by practitioners due to its asymmetry property: focus on the maximum drawdown not on volatility like in $SR$. 
\begin{equation*}
  CR = \frac{{252\times AV}}{{D_{P}}},
\end{equation*}
where $AV$ is the average out-of-sample return multiplied by 252 to annulize and $D_P$ is the maximum drawdown of the portfolio returns.

 \end{itemize}

Since the focus of this research is the reduction of portfolio weights' fluctuations, following \cite{ledoit2017nonlinear} we also compute the following five characteristics of weights' vectors $\hat w_{t}$ averaged through number of rebalancing periods. Thus, we calculate \emph{minimum weight (min)} for every benchmark strategy as follows:

 \begin{equation*}
   \text{min} =   \frac{1}{T_{\max}-T}{\sum_{t=T}^{T_{\max}}}\min(\hat{w}_{t})   \, .
  \end{equation*}
 We similarly compute maximum weight \emph{(max)},  standard deviation \emph{(sd)}, and  range of weights~\emph{(max-min)}.

In addition, we provide MAD from EW portfolio \emph{(mad-ew)}, which is defined as:

  \begin{equation*}
   \text{mad-ew}  =  \frac{1}{T_{\max}-T}\sum_{t=T}^{T_{\max}}{N}^{-1}\sum_{j=1}^{N}\left|\widehat{w}_{j,t}-\frac{1}{N}\right| \, . 
  \end{equation*}

\section{Empirical study}
\label{sec:emp}

\subsubsection*{Discussion of S\&P100  data set results}
First, we discuss portfolio weight stability, since it is the main focus of the research. \reffig{Weights_SP} demonstrates the dynamics of weights for S\&P100. It can be observed that weights of plug-in GMV portfolio are characterized by a lot of extremes in comparison with all other policies. The least dispersed weights are observed for, introduced in this paper, GMV\_{robust} approach. This visual result is confirmed by descriptive statistics of portfolio weights reported in the \reftab{Weights_out-of-sampleSP100}.  It can be found that the average range of weights for GMV\_{robust} 0.07 is the lowest one, what is almost four times less than the range of GMV\_lin, GMV\_nlin and GMV\_long and five times less than plug-in GMV policy; $mad-ew$ also is the lowest for GMV\_{robust} strategy, pointing out the more balanced distribution of weights. \reftab{Performance_SP1002} reports these results, which can be summarized as follows. 

The two main characteristics of interest for this research would be \emph{Turnover} and \emph{Target Turnover}. As expected, the best performing policy in this dimension is GMV\_long with imposed non-negative constraints; it requires on average almost 13\% (TO) of trading volume to rebalance the portfolio. GMV\_{long} is followed by EW with 20\% and GMV\_{robust} with 25\%. The highest turnover is reached by GMV with almost 71\% of portfolio value to rebalance the portfolio to target weights $\hat{w_t}$. The ranking of strategies via $TTO$ stays unchanged. The sample conducted t-test, confirms significant difference of GMV\_{robust} $TO$ and $TTO$ in comparison with other benchmarks.

Naturally, for investors, Cumulative Wealth ($CW$) is of high interest as a measure of performance for the period considered. The gross $CW$  for portfolio benchmarks does not vary considerably, ranging from 256\% for GMV\_{long} portfolio to 273\% for EW. However, the inclusion of transaction fees changes the  results:  the best performing EW portfolio gains 249\%, followed by GMV\_{robust} with CW 242\%. The worst performing portfolio is GMV with 188\% of CW. The evolvement of gross CW of all benchmark strategies for the considered period is plotted in \reffig{CumWealthRussell3000}. 

Shrinkage estimators exhibit the lowest risk, measured by SD and  tests conducted point out the significance of this difference with GMV\_robust. In terms of risk-adjusted performance, the winning strategies for Gross returns  are shrinkage portfolios and GMV with SR 5.71\% and 5.65\% respectively. However, for net returns they lose half of their performance, ending up with 2.3\%, whereas SR of our approach stays more stable changing from 4.26\% to 3.2\%. In addition, conducted tests of difference significance  support the hypothesis of equal SRs, although it can be noticed that the difference is becoming more prominent for net returns. 

Finally, GMV\_{robust} dominates other benchmarks in terms of Calmar ratio. The drawdowns often determine whether a buy-side investor can keep an investment or will have to unwind and thus miss subsequent recoveries, what makes it an important risk measure for buy-side investors and managers. 

\subsubsection*{Discussion of Russell3000  data set results}
Outcomes of weights stability analysis are consistent with ones described for S\&P100 data set. GMV\_{robust} weights are characterized by harmonized weights without extreme short or long values. It is visible in the \reffig{Weights_Russell} and in the \reftab{Weights_out-of-sampleRussell3000}: GMV\_{robust} $MAD-EW$ and $max-min$ are the lowest in comparison with benchmark portfolios (excluding EW). 

In terms of accumulated wealth, GMV\_{robust} for large portfolios performs very close to shrinkage estimators, \reftab{Performance_Russell30002} summarizes investment performance characteristics. 
Thus, GMV\_{robust} gains in the end of the period 210 \% of initial value while GMV\_lin and GMV\_nlin 209\% and 201.9\%. But considering non-zero trading fees would change the rank drastically, making the best performing benchmark Index itself, followed by GMV\_{robust}. The findings for other risk measures stay in line with findings for S\&P100: 
\begin{itemize}
    \item  the less risky portfolios with the lowest SD are shrinkage portfolios, and the difference is significant in relation to  GMV\_{robust}.
 \item Shrinkage estimators demonstrate the highest SR for gross returns, but for the net returns  GMV\_{robust} overperforms, although difference tests do not prove the significance of the result. 
 \item GMV\_{robust} dominates other strategies in terms of  CR for both net and gross returns.
  \end{itemize} 

Thus, according to outcomes of empirical experiments we can claim that GMV\_{robust} portfolio policy achieves its goal and substantially reduces fluctuations of weights, leading to the lowest level of accumulated transaction costs. The risk adjusted performance is equal or slightly lower than shrinkage benchmarks and higher than constrained rules. Our approach is of specific interest for invetors and managers who focus on the worst drawdowm. This conclusion stays robust for small and large portfolios. 

\tab{Performance_SP1002}{\footnotesize 
{ 
Out-of-sample performance measures of benchmark portfolios, 81 stocks of S\&P100, monthly rebalancing: TTO, target turnover; TO, Turnover;  CW, cumulative Wealth; SD, standart deviation; SR, Sharpe ratio; CR, Calmar ratio. All estimates are obtained from daily values (4781 out-of-sample returns) for both gross and net returns. Time period: 20020101 - 20201231. The difference tests for the turnovers are obtained from a sample $t$-test. The difference test on the Sharpe ratio and Variance used the approach of \cite{ledoit2008robust} and \cite{ledoit2011robust}, for which the R code is available at M. Wolf's website.
$p$-values are reported in parentheses with respect to the GMV\_{robust} portfolio.} \quantlet{RobustM_PerformanceSP100}{RobustM\_PerformanceSP100} }

\tabf{Weights_out-of-sampleSP100}{\footnotesize{Average characteristics of the weight vectors of   GMV portfolios, 81 stocks of S\&P100, monthly rebalancing. Time period: 20020101 - 20201231}
\quantlet{RobustM_PerformanceSP100}{RobustM\_PerformanceSP100}}
\tabf{Performance_Russell30002}{\footnotesize 
{ Out-of-sample performance measures of benchmark portfolios, 600 stocks of Russell3000, monthly rebalancing: TTO, target turnover; TO, Turnover, TC, transaction costs;  CW, cumulative Wealth; SD, standart deviation; SR, Sharpe ratio; CR, Calmar ratio. All estimates are obtained from daily values (2265 out-of-sample returns) for both gross and net returns. Time period: 20120101 - 20201231. The difference tests for  turnovers are obtained from a sample $t$-test. The difference test on the Sharpe ratio and Variance used the approach of \cite{ledoit2008robust} and \cite{ledoit2011robust}, for which the R code is available at  M. Wolf's website. $p$-values are reported in parentheses with respect to the GMV\_{robust} portfolio.} 
\quantlet{RobustM_PerformanceRussell3000}{RobustM\_PerformanceRussell3000}}
\begin{figure}
    \centering
    \subfloat[\centering S\&P100 dataset]{{\includegraphics[width=6.5cm]{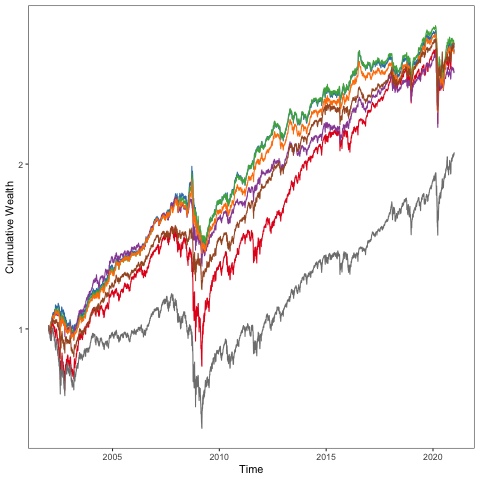} }}%
    \qquad
    \subfloat[\centering Russell3000 dataset]{{\includegraphics[width=6.5cm]{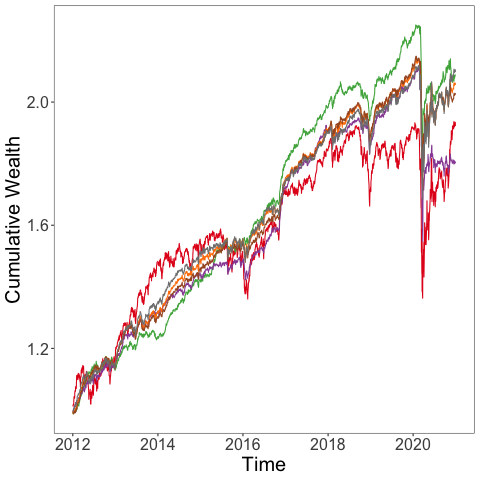} }}%
    \caption{\footnotesize {Cumulative wealth of benchmark portfolios \textcolor{EW}{EW}, 
\textcolor{GMV_lin}{GMV\_lin}, \textcolor{GMV_long}{GMV\_long}, \textcolor{GMV_nlin}{GMV\_nlin}, \textcolor{GMV_{robust}}{GMV\_{robust}}, \textcolor{Index}{Index}. Panel (a): 81 constituents of S\&P100, time period 20020101 - 20201231 (4781 daily returns);  panel (b)) 600 random constituents of Russell3000, time period 20120101 - 20201231 (2263 daily returns)}}
    \label{CumWealthRussell3000} \quantlet{RobustM_PerformanceRussell3000}{RobustM\_PerformanceRussell3000}
\end{figure}
\tab{Weights_out-of-sampleRussell3000}{Average characteristics of the weight vectors  of  GMV portfolios, 600 stocks of Russell3000, monthly rebalancing. Time period: 20020101 - 20201231 \quantlet{RobustM_PerformanceRussell3000}{RobustM\_PerformanceRussell3000}}

\section{Conclusion and discussion}

``Robustifying Markowitz'' has seen many attempts that are mostly based on robustifying the original simple inversion formula for exact determination of optimal GMV weights. In bypassing this ``error maximizing'' technique, we have presented a tool fixing the portfolio weights in a low cost re-balancing 
ballpark.  Using modern results from robust statistics, we have constructed an algorithm that provides a  computationally effective estimator for GMV Markowitz portfolios. We have shown that it suffices to utilize a PGD procedure to optimize the portfolio weights without estimating the covariance operator itself.  The focus on just the PGD updates significantly distinguishes our approach from the previous techniques.  
We have successfully derived almost Gaussian properties of this estimator in nice ($N/T$ small) and not-so-nice ($N/T$ big) condition cases. 

The weights developed with the robustified approach are less sensitive to deviations of the asset-return distribution from normality than those of the traditional minimum-variance policy. Empirical studies confirm that the proposed policies are indeed more stable and cost-reducing. The stability of the proposed portfolios makes them a feasible alternative to traditional portfolios.

The proposed toolbox improves the stability properties of weights, leading to better investment characteristics of allocation policies. Although inferior at the risk level measured by SD, our algorithm reaches a superior risk-adjusted performance (Sharpe and Calmar ratios) for net returns due to a substantial reduction of trading volume measured by turnover. Finally, these performance results are confirmed across small and large portfolios.  Even for dimensions of portfolio size larger than the length of the estimation window (e.g., the Russell3000 data), the above claim pertains.

\section*{Acknowledgments}

Wolfgang H\"ardle and Alla Petukhina gratefully acknowledge the financial support of the European Union's Horizon 2020 research and innovation program ``FIN-TECH: A Financial supervision and Technology compliance training programme" under the grant agreement No 825215 (Topic: ICT-35-2018, Type of action: CSA), the European Cooperation in Science \& Technology COST Action grant CA19130 - Fintech and Artificial Intelligence in Finance - Towards a transparent financial industry, the Deutsche Forschungsgemeinschaft's IRTG 1792 grant, Wolfgang H\"ardle - the Yushan Scholar Program of Taiwan, the Czech Science Foundation's grant no. 19-28231X / CAS: XDA 23020303. Nikita Zhivotovskiy was funded in part by ETH Foundations of Data Science (ETH-FDS).  We wish to thank three anonymous referees for their thoughtful comments and efforts toward improving our manuscript. We also gratefully acknowledge the comments and discussions from Valerio Poti, Natalie Packham, and Jörg Osterrieder, as well as participants of COST FinAI Annual Meeting in Bucharest, Research seminars "Mathematics" at HTW Berlin and the Department of Economics and Management, University of Pavia. We appreciate the editorial assistance of Elie Tamer and Terry Liu.

\bibliography{references}

\appendix

\section{Simulation study}\label{sec:sim}

In Section~\ref{sec:ill_cond}, we have shown the convergence properties of our iterative algorithm depending on the number of steps taken. Here, we use simulations to study these properties empirically. In addition, we demonstrate numerically that using a robust estimator of $ \Sigma w $ improves the confidence bounds for the GMV risk.

\subsubsection*{Convergence of the algorithm}

We first generate $T = 250$ independent observations from Gaussian distribution. To get closer to the real distribution, we use a covariance matrix estimated from the real dataset. 

We take a 250 days window from 2012-08-27 to 2013-08-26 of $N=82$ constituents of S\&P100, and compute their sample covariance denoted as $S_1$. The corresponding optimal GMV weights $ w_1 = (S_1^{-1} \mathbf{1}) / (\mathbf{1}^{\T} S_1^{-1} \mathbf{1}) $ are rather explosive and have the maximal absolute weight $0.24$, which is $ \approx 20 $ times bigger than for the equally weighted portfolio.

It is arguable whether the true optimal portfolio weights are explosive in reality. We additionally consider a modified covariance matrix, that is obtained by rotating the eigenbasis with minimal efforts, so that the vector of ones $ \mathbf{1} = (1, 1, \dots, 1)^{\T} $ belongs to the span of top $15$ eigenvectors. By taking $ S_2 = R S_1 R^{\T} $, where $R$ is a rotation matrix, we ensure that $ \| S_1 - S_2 \| / \| S_1\| \approx 0.13 $, while $ S_2^{-1} \mathbf{1} $ is not as explosive, and so are the optimal GMV weights $ w_2 = (S_2^{-1} \mathbf{1}) / (\mathbf{1}^{\T} S_2^{-1} \mathbf{1}) $. To be precise, the absolute maximal weight becomes $ 0.08 $, which is 3x smaller than for the original sample covariance.

\begin{figure}
     \centering
     \begin{subfigure}[b]{0.3\textwidth}
         \centering
         \includegraphics[width=\textwidth]{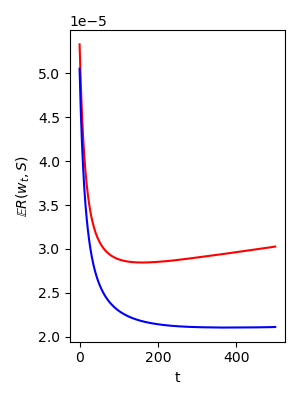}
         \caption{}
         \label{fig:convsim1}
     \end{subfigure}
     \hfill
     \begin{subfigure}[b]{0.3\textwidth}
         \centering
         \includegraphics[width=\textwidth]{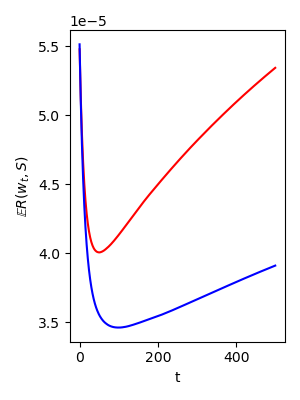}
         \caption{}
         \label{fig:convsim2}
     \end{subfigure}
     \hfill
     \begin{subfigure}[b]{0.3\textwidth}
         \centering
         \includegraphics[width=\textwidth]{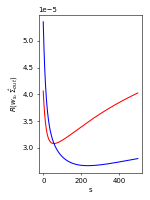}
         \caption{}
         \label{fig:convsim3}
     \end{subfigure}
        \caption{The graphs (a), (b) show simulated expected risk $ \E R(w_s, \Sigma) $ (red), expected in-sample risk $ \E R(w_s, \hat{\Sigma}) $ (blue) for $ \Sigma = S_1 $ (a) and $ \Sigma = S_2 $ (b). The graph (c) shows in-sample risk (blue) for the S\&P100 data between 2012-08-27 and 2013-08-26, and the out-of-sample risk (red) estimated using the time window between 2013-08-27 and 2014-01-17.}
        \label{fig:convsim}
        \quantlet{RobustM_Simulations}{RobustM\_Simulations} 
\end{figure}

In Figures~\ref{fig:convsim1}--\ref{fig:convsim3} the blue line shows the (expected) in-sample risk of $w_s$, the $s$th iteration of the algorithm, depending on the number of steps. The red line shows the population or out-of-sample risk. 
        In Figures~\ref{fig:convsim1} and \ref{fig:convsim2}, the generated data has Gaussian distribution with covariance $S_1$ and $S_2$, respectively (we take zero mean, since the estimator is translation invariant, and the risk measure does not depend on it). Notice that although the in-sample risk generally decreases with a growing number of steps, the population risk $ \E R(w_s, S_j) $ has a single minimum, which confirms the tradeoff stated in Lemma~\ref{lemma_ill_cond}, where a specific number of steps needs to be taken to achieve the bound. Using Figure~\ref{fig:convsim3} we can compare the simulated data with real, where the blue line represents one realization of the in-sample risk, using the data from 2012-08-27 to 2013-08-26, and the red line shows the out-of-sample risk that is calculated using the data from 2013-08-27 to 2014-01-17. This graph shows similar features as the first two. Notice that the curve is steeper in cases (b) and (c). Motivated by this we will use the covariance matrix $S_2$ for further simulations.

\subsubsection*{Heavy tailed simulation}

We further investigate the effect of using a robust estimator with heavy-tailed data. In particular, we want to compare the performance of a robust estimator against the classical one. Recall, that our algorithm consists of iterative updates,
\[
    w_0 = \mathbf{1} / N,
    \qquad
    w_{s} = \Pi_{1} [w_{s - 1} - \eta \hat{a}(w_{s - 1})],
\]
where $\Pi_{1}$ is defined in \eqref{Pi_1_definition}, and $ \hat{a}(w) $ is a robust estimator of $ \Sigma w $. What happens if we replace it with the standard empirical version $ \hat{\Sigma} w$? Below we demonstrate that using a robust estimator improves the high probability bounds.

For simulations, we construct a heavy-tailed distribution of the form $ X = S_2^{1/2} Y $, where
\begin{equation}
\label{heavy_mixture}
    Y \sim 0.999 \times \mathcal{N}(0, \mathbf{I}) + 0.001 \times D,
\end{equation}
i.e., with probability $0.999$, $Y$ is sampled from a standard normal distribution and, with a small probability $0.001$, from a heavy-tailed distribution $D$. In our example, $D$ is supported in $ \{ -1, 1\}^{82}$. To sample a vector $ Z \sim D$, we first take a random integer $k $ uniformly in the set $\{ 0, 1, \dots, N=82\}$. Then, we take a random subset $B$ of $ \{ 1, \dots, 82\} $ of size $k$, uniformly over all such subsets. Then, we assign $ Z_i = 1 $ for $i \in B$ and $Z_i = -1$ otherwise. Due to the symmetry, we have $ \E Z = 0 $, and it is straightforward to calculate that $ \mathrm{Var}(Z) = \tfrac{2}{3} \mathbf{I} + \tfrac{1}{3} \mathbf{1}\mathbf{1}^{\T} $. Then, the covariance matrix of $ X = S_2^{1/2} Y $ is as follows,
\[
    \mathrm{Var}(X) = S_2^{1/2} \left[ \tfrac{1}{3} 2.999 \times \mathbf{I} + \tfrac{1}{3} 0.001 \times \mathbf{1}\mathbf{1}^{\T} \right] S_2^{1/2} \, .
\]
The eigenvalues of the matrix in the middle are in the range $[0.999, 1.055]$, so that $\mathrm{Var}(X)$ is only a slight perturbation of the original covariance $S_2$.

\begin{figure}
     \centering
     \begin{subfigure}[b]{0.3\textwidth}
         \centering
         \includegraphics[width=\textwidth]{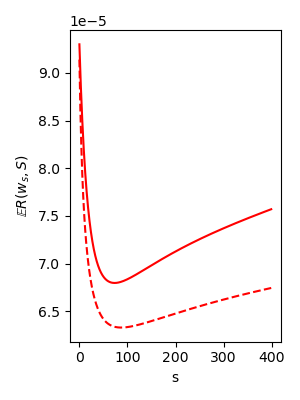}
         \caption{}
         \label{fig:heavy_mean}
     \end{subfigure}
     \hfill
     \begin{subfigure}[b]{0.3\textwidth}
         \centering
         \includegraphics[width=\textwidth]{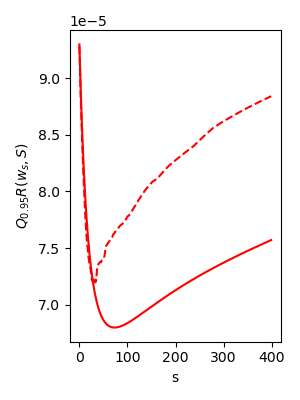}
         \caption{}
         \label{fig:heavy_95}
     \end{subfigure}
     \hfill
     \begin{subfigure}[b]{0.3\textwidth}
         \centering
         \includegraphics[width=\textwidth]{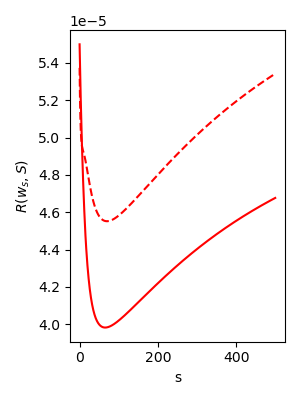}
         \caption{}
         \label{fig:contamination}
     \end{subfigure}
        \caption{Graph (a) shows  expected risk $ \E R(w_s, \Sigma) $  for the robust and non-robust (dashed) estimators. Graph (b) shows 95\% quantile of the risk in the same setting. Graph (c) shows the risk of robust and non-robust estimators with one contaminated observation.}
        \label{fig:heavy_all}
\quantlet{RobustM_Simulations}{RobustM\_Simulations} 
\end{figure}

We generate $200$ independent samples of size $250$ from this distribution. For each sample, we compute the true risk of $s$-th iteration, given that we know exactly the covariance matrix. Figure~\ref{fig:heavy_mean} shows the average risk of $w_s$ depending on the number of iterations, and Figure~\ref{fig:heavy_mean} shows its $95\%$ quantile, which are calculated by taking the mean and percentile of these $200$ independent simulations, respectively. The thick line represents the risk of the robust procedure, while the dashed line corresponds to the non-robust version of the algorithm with $ \hat{a}(w) $ replaced by $ \hat{\Sigma} w $. It is evident, that while there is no improvement on average, the robust estimator has significantly better tails. In addition, Figure~\ref{fig:contamination} shows what happens when we treat the heavy part of the distribution as contamination. We generate a single realization of a Gaussian sample and replace the first element with a random realization sampled from ~$D$. In this case, we calculate the risk using the original covariance matrix~$S_2$.

\section{Proofs}

\subsection{Proof of Proposition~\ref{a_hat_prop}}
\label{sec:proof_a_hat}

Before we proceed, let us recall some of the results and definitions from \citet{lugosi2019sub} and \citet{hopkins2020robust}. We start by giving the definition of {a combinatorial center. It is the central object in the original construction of \citeauthor{lugosi2019sub}, but the definition itself is due to \citeauthor{hopkins2020robust}}

\begin{definition}[Combinatorial center]
A point $\theta \in \R^{N}$ is called {an} $(r, \kappa)$-combinatorial center of $Y_1, \dots, Y_{\ell}$ if for all unit vectors $ v \in \R^{N}$, the inequality
\[
    |v^{\T}(Y_j - \theta)| \leq r
\]
takes place for at least $(1 - \kappa)\ell$ of indices $j = 1, \dots, \ell$.
\end{definition}

Essentially, \cite{lugosi2019sub} prove that for appropriately chosen $r_{\delta}$, the true mean is an $(r_{\delta}, 1/4)$-combinatorial center, with probability at least $1 - \delta$, where $Y_j$ are the bucket means (see their paper for exact definitions). The estimation strategy is then executed by what is called a \emph{median-of-means tournament}: one needs to pick an $(r, 1/4)$-combinatorial center with $r$ as small as possible. The deviation bound then follows by a simple triangle inequality. One difficulty of implementing this strategy computationally is the lack of control on how these subsets of indices behave for different directions $v \in \R^{N}$.

In addition, \cite{hopkins2020robust}  define the \emph{spectral center} as follows.

\begin{definition}[Spectral center]
For $\varepsilon \in (0, 1/2)$, denote 
$$ 
    \Delta_{\ell, \varepsilon} = \left\{ u \in \R^{\ell}: \; \sum_{j = 1}^{\ell} u_{j} = 1,\;\; 0 \leq u_{j} \leq 1 / \{\ell(1 - \varepsilon)\} \right\} \, .
$$
A point $\theta \in \R^{N}$ is called an $(r, \varepsilon)$-spectral center if there are weights $(u_1, \dots, u_{\ell}) \in \Delta_{\ell, \varepsilon}$ such that 
\[
    \left\| \sum_{j = 1}^{\ell} u_j (Y_j - \theta)(Y_j - \theta)^{\T} \right\| \leq r^2 .
\]
\end{definition}

It is straightforward to see that if $\theta$ is a $(r, \varepsilon)$-spectral center with minimal $r$, then it has the form
$
    \theta = \sum_{j = 1}^{\ell} u_{j} Y_{j}
$
for some $(u_1, \dots, u_{\ell}) \in \Delta_{\ell, \varepsilon}$, i.e., the solution should be a weighted mean of $Y_j$. The two definitions are ``equivalent'' in the  following sense.

\begin{lemma}\label{lem:spectr_comb}
Suppose that $\theta$ is an $(r, \kappa)$-combinatorial center. Then it is also a $(5r, 10\kappa)$-spectral center. Conversely, if $\theta$ is an $(r, \varepsilon)$-spectral center, then it is also a $(\sqrt{(1-\varepsilon)/\varepsilon} \, r, 2\varepsilon)$-combinatorial center.
\end{lemma}

\cite{hopkins2020robust} state this lemma for some particular constants $\varepsilon$ and $\kappa$. Their proof consists of some arguments of the proof of Proposition~1 in \citet{depersin2019robust}. For the sake of completeness, we  reproduce these arguments in Section~\ref{sec:proof_spectr_comb}, slightly changed.

We deal with both notions of centers for the following reason: it is easier to deal with the statistical properties of combinatorial centers, whereas the spectral centers are more convenient from computational perspective. \cite*{hopkins2020robust} develop an algorithm, namely their Algorithm $1$, that finds a spectral center with the radius that is guaranteed at most (say) twice as large as minimal possible. Namely, we denote the output of their algorithm as $  \mathsf{HLZ}(Y_1, \dots, Y_{\ell}; \varepsilon)$ and they show that the output $\hat{\mu}$  satisfies
\begin{equation}\label{hopkins_guarantee}
    \min_{u \in \Delta_{\ell, \varepsilon}} \left\| \sum_{j = 1}^{\ell} u_{j} (Y_j - \hat{\mu})(Y_j - \hat{\mu})^{\T} \right\| \lesssim \min_{\theta} \min_{u \in \Delta_{\ell, \varepsilon}} \left\| \sum_{j = 1}^{\ell} u_{j} (Y_j - \theta)(Y_j - \theta)^{\T} \right\| \, .
\end{equation}
Hence our goal is to show that with high probability, the true mean is an $(r, \varepsilon)$-spectral center with some fixed $\varepsilon$ and sufficiently small $r$, which we can do using the median-of-means analysis and switching back and forth between spectral and combinatorial centers.

Let us now give the description of the estimator $ \hat{a}_{\delta}(w) $. It consists of the following steps:
\begin{enumerate}
    \item First we  centralize our observations. For this, consider the  transformations $ \widetilde{X}_1 = (X_1 - X_2) / \sqrt{2}, \widetilde{X}_2 = (X_3 - X_4) / \sqrt{2}, \dots $ Obviously, each of these new ``observations'' has zero mean and the same covariance as $X_i$, and moreover they are independent.
    \item Fix $ \varepsilon < 10/21$ and set $ \ell = \lceil 2(\varepsilon / 10)^{-2} \log (2/\delta) \rceil  $. Split the observations $ \widetilde{X}_1, \dots, \widetilde{X}_{\lfloor T/2 \rfloor} $ into $\ell$ non-intersecting buckets \[
        B_1 \sqcup \dots \sqcup B_{\ell} = \{1, \dots, \lfloor T/2\rfloor \} .
    \]
    \item Next, using the data from each of the buckets, we construct the following covariance estimators,
    \[
        \Sigma_{j} = \frac{1}{|B_j|} \sum_{i \in B_{j}} \widetilde{X}_i \widetilde{X}_i \mathbf{1}_{\| \widetilde{X}_i\| \leq D}.
    \]
    \item For a given direction $w \in \R^{N}$, we output the result of the $\mathsf{HLZ}$ algorithm applied to the bucket means
    \[
        \hat{a}_{\delta}(w) \myeq \mathsf{HLZ}(\Sigma_1 w, \dots, \Sigma_{\ell} w; \varepsilon) \, .
    \]
\end{enumerate}

\begin{remark}
Notice that given $ \varepsilon < 10/21$ and fixing, say, $\delta = 0.05$, we have that the number of buckets $ \ell = \lceil 2 (\varepsilon / 10)^{-2} \log(2/\delta) \rceil $ has to be at least $ 1500$, which makes the algorithm rather impractical. Unfortunately, our techniques do not allow more adequate constants. In the empirical study we heuristically choose $ \varepsilon = 1/3 $ and $ \ell = 10 $.
\end{remark}

By considering $ (X_1 - X_2) / \sqrt{2}, (X_3 - X_4) / \sqrt{2}, \dots $ we guarantee that our observations are zero mean without changing the covariance matrix. We have done so by reducing the size of the sample by at most two. To avoid the notation overloading, we assume below that $\mu = 0$ and proceed to work with the original $X_i$.

Set $ Z_i = {X}_i {X}_i^{\T} \mathbf{1}[\| {X}_i\| < D]$.
According to steps 2 and 3 of the algorithm, we split the data into $ \ell $ blocks $ B_1, \dots, B_{\ell} $ and consider the trimmed covariances,
\[
    {\Sigma}_{j} = \frac{1}{|B_j|} \sum_{i \in B_j} Z_i \, .
\]
Below we derive the following bound: with probability $1-\delta$, we have that for any directions $v, w$, the inequality
\begin{equation}\label{unifrom_sigma_mom}
    | u^{\T} {\Sigma}_{j} w - u^{\T} \Sigma w | \lesssim  \Delta_{\Sigma}\myeq \| \Sigma \| \sqrt{\frac{ \mathbf{r}(\Sigma) \log \mathbf{r}(\Sigma) +  \ell}{T}} 
\end{equation}
holds for at least $1 - \kappa$ fraction of the indices $j = 1, \dots, \ell$, where $ \kappa = \varepsilon / 10 $ is fixed. Let us first complete the proof given this inequality.

At Step 5 of the algorithm, we produce the vectors $Y_j = \Sigma_j w $. On the event from \eqref{unifrom_sigma_mom}, we have that for any unit $u \in \R^{N}$, the inequality
\[
    |u^{\T} (Y_j - \Sigma w) | \leq C \Delta_{\Sigma}
\]
holds for at least $ 1-\kappa $ a fraction of indices. Hence, $\Sigma w$ is a $(C \Delta_{\Sigma}, \kappa)$-combinatorial center of $Y_j$. By Lemma~\ref{lem:spectr_comb}, it also means that $\Sigma w$ is a $(5C \Delta_{\Sigma}, 10\kappa)$-spectral center of $Y_j$. Hence, by \eqref{hopkins_guarantee}, the output of $\hat{a}_{\delta}(w) = \mathsf{HLZ}(Y_1, \dots, Y_j; \varepsilon)$ is a $(C' \Delta_{\Sigma}, 10\kappa)$-spectral center, and using the second part of Lemma~\ref{lem:spectr_comb}, we conclude that it is also a $(C' \sqrt{10/\kappa} \Delta_{\Sigma}, 20\kappa)$-combinatorial center. Since $ 21\kappa < 1 $, we get that $ 1 - 20 \kappa + 1 - \kappa > 1$, which means that for any direction $ u \in\R^{d} $, by the pigeonhole principle, we can pick a single $Y_j$ that is close to both combinatorial centers $\Sigma w$ and $ \hat{a}_{\delta}(w) $ in this direction. Therefore, by the triangle inequality,
\[
    | u^{\T} (\hat{a}_{\delta}(w) - \Sigma w) | \lesssim \Delta_{\Sigma},
\]
and since the bound holds in arbitrary direction $u$, we get the required bound in Euclidean norm. 

It remains to prove the bound \eqref{unifrom_sigma_mom}. Let $ \widetilde{\Sigma} = \E Z_i$. We have by Lemma~2.1 of \cite{mendelson2020robust},
\begin{equation}\label{sigma cutoff error}
    \| \Sigma  - \widetilde{\Sigma}\| \lesssim \frac{\| \Sigma \|^2 \mathbf{r}(\Sigma)}{D^2} \, .
\end{equation}
Let $ {Quant}_{\alpha}(z_1 \dots z_\ell)$ of a sequence of real numbers denotes an order statistics $z_{(\lceil \alpha \ell \rceil)}$, where $z_{(1)} \dots z_{(k)}$ is a non-decreasing rearrangement of the original sequence. Then, we can rewrite \eqref{unifrom_sigma_mom} as follows,
\[
    \max\left\{ Quant_{1 - \kappa}(u^{\T} \Sigma_j w) - u^{\T} \Sigma w, u^{\T} \Sigma w - Quant_{\kappa}(u^{\T} \Sigma_j w) \right\} \lesssim \Delta_{\Sigma}\, .
\]
Let us apply \citet[Lemma~2.3]{klochkov2020robust} to the class of functions $ \{ f_{u, w}(Y) = u^{\T} Y w \}$. We have that with probability at least $1 - 2e^{-\kappa^2 \ell / 2}$,
\begin{multline*}
    \max\left\{ Quant_{1 - \kappa}(u^{\T} \Sigma_j w) - u^{\T} \Sigma w, u^{\T} \Sigma w - Quant_{\kappa}(u^{\T} \Sigma_j w) \right\} \\
    \lesssim \E \sup_{u, w} \left( \frac{1}{T} \sum_{i= 1}^{N} \varepsilon_{i} u^{\T} Y_i w \right) + \sqrt{\sup_{u, w} \E (u^{\T} Y_1 w)^{2} \frac{\ell}{T}} + \| \widetilde{\Sigma} - \Sigma \| \, ,
\end{multline*}
where $ \varepsilon_1, \dots, \varepsilon_n$ are independent Rademacher signs, i.e., taking $\pm 1$ with probability $1/2$.  The supremum in both terms of the RHS is over unit vectors $u, w$ in $\R^{N}$. Let us first bound the second, weak term.
We have that
\[
    \E (u^{\T} Y_1 w)^{2} \leq \E (u^{\T} X_1)^2 (w^{\T} X_1)^2 \leq \E^{1/2} (u^{\T} X_1)^4 \E^{1/2} (w^{\T} X_1)^4 \, .  
\]
By the $L_4$--$L_2$
equivalence assumption we get that $ \E^{1/2} (u^{\T} X_1)^{4} \lesssim \E (u^{\T} X_1)^{2} \leq \| \Sigma \| $.
The weak term is therefore bounded by $ C \| \Sigma \| \sqrt{\ell / T} $.

Now let us deal with the first, strong term. We rewrite it as follows,
\[
    \E \sup_{u, w} \left( \frac{1}{T} \sum_{i= 1}^{N} \varepsilon_{i} u^{\T} Y_i w \right) = \E \sup_{u, w} u^{\T} \left( \frac{1}{T} \sum_{i= 1}^{N} \varepsilon_{i} Y_i \right) w = \E \left\| \frac{1}{T} \sum_{i= 1}^{N} \varepsilon_{i} Y_i \right\|.
\]
The right-most expression is the expected value of the norm of a sum of centered matrices $\varepsilon_i Y_i$, which are bounded by $D^2$ pointwise.
We therefore can apply the Matrix Bernstein inequality, the details are carried out by \cite{mendelson2020robust}. They show that this leads eventually to the bound
\[
    \E \left\| \frac{1}{T} \sum_{i = 1}^{T} \varepsilon_i Y_i \right\| \lesssim \| \Sigma \| \sqrt{\frac{\mathbf{r}(\Sigma) \log \mathbf{r} (\Sigma)}{T}} + \frac{D^{2} \log \mathbf{r}(\Sigma)}{T} \, .
\]
Recalling the bound \eqref{sigma cutoff error} we get that, with probability $1 - 2e^{-\kappa^2 \ell / 2}$,
\begin{multline*}
    \max\left\{ Quant_{1 - \kappa}(u^{\T} \Sigma_j w) - u^{\T} \Sigma w, u^{\T} \Sigma w - Quant_{\kappa}(u^{\T} \Sigma_j w) \right\} \\
    \lesssim \Delta_{\Sigma} + \frac{D^{2} \log \mathbf{r}(\Sigma)}{T} + \frac{\| \Sigma \|^{2} \mathbf{r}(\Sigma)}{D^{2}} \, .
\end{multline*}
For $ D \sim \| \Sigma\|^{1/2} \left( \frac{T \mathbf{r}(\Sigma)}{\log \mathbf{r}(\Sigma)}\right)^{1/4}  $ the RHS simplifies to $\Delta_{\Sigma} = \| \Sigma \| \sqrt{\frac{\mathbf{r}(\Sigma) \log \mathbf{r}(\Sigma) + \ell}{T}}$. It remains to notice that $ 1-2e^{-\kappa^2 \ell / 2} \geq 1 - \delta $ as long as $ \ell \geq  2 \kappa^{-2} \log\left(\frac{2}{\delta}\right) $.

\subsection{Proof of Lemma~\ref{lem:spectr_comb}}\label{sec:proof_spectr_comb}

Let us first recall the following basic fact from linear algebra: for a symmetric matrix $A$, its largest eigenvalue satisfies $ \lambda_{\max}(A) = \sup_{M \succeq 0, \Tr(M) = 1} \Tr(MA) $. Hence, we can rewrite
\[
    \min_{w \in \Delta_{\ell, \varepsilon}} \left\| \sum_{j = 1}^{\ell} w_j (Y_j - \theta)(Y_j - \theta)^{\T} \right\| = \min_{w \in \Delta_{\ell, \varepsilon}} \max_{M \succeq 0, \Tr(M) = 1} \sum_{j = 1}^{\ell} w_j (Y_j - \theta)^{\T} M (Y_j - \theta).
\]
The latter can be seen as \emph{semi-definite program} (SDP) and using the strong duality of SDP one can show that the minimum over $w$ and the maximum over $M$ can be swapped (see formula (5.2) in \citet{hopkins2020robust}; see also \citet{depersin2019robust, diakonikolas2020outlier}):
\begin{align*}
    &\min_{w \in \Delta_{\ell, \varepsilon}} \max_{M \succeq 0, \Tr(M) = 1} \sum_{j = 1}^{\ell} w_j (Y_j - \theta)^{\T} M (Y_j - \theta)
    \\
    &= \max_{M \succeq 0, \Tr(M) = 1} \min_{w \in \Delta_{\ell, \varepsilon}} \sum_{j = 1}^{\ell} w_j (Y_j - \theta)^{\T} M (Y_j - \theta).
\end{align*}
The right-hand side form is closer to what \cite{lugosi2019sub} do: for any direction $M = vv^{\T}$ we can pick its own weights. This property allows to show the equivalence.

We write $y_j = Y_j - \theta$ for short everywhere in this section. First, assume that $\theta$ is a $(r, \kappa)$-combinatorial center. We will show by contradiction that it is also a $(\tilde{r}, \varepsilon)$-spectral center, where $\tilde{r} = 5r$ and $\varepsilon = 10\kappa$. Suppose it is not, so that for some $M \succeq 0$ with $\Tr(M) = 1$ we have that
\[
    \min_{w \in \Delta_{\ell, \varepsilon}} \sum_{j} w_j y_j^{\T} M y_j \geq \tilde{r}^2 \, .
\]
If $w \in \Delta_{\ell, \varepsilon}$ delivers the minimum it must put non-zero weights to at least $\lceil \ell(1 - \varepsilon) \rceil $ terms. Since the weights sum up to one, we conclude that for at least $\lceil \ell \varepsilon \rceil$ indices $j = 1, \dots, \ell$, it holds that $y_j^{\T} M y_j \geq \tilde{r}^2 $. We denote this set of indices as $B$. Now, let $  M = \sum_{k = 1}^{\ell} \lambda_k u_k u_k^{\T} $ be its spectral decomposition. Since $ M\succeq 0 $ and $\Tr(M) = 1$, we have that $ \sum_{k} \lambda_{k} = 1  $ and $\lambda_k \geq 0$. 

Let us take a random unit vector $ v = \sum_{k} \sqrt{\lambda_k} u_k \varepsilon_{k} $, where $\varepsilon_{k}$ are independent random signs, so that the equality $\sum_{k} \lambda_k = 1$ ensures that it is indeed a unit vector. Moreover,
\[
    y_j^{\T} v = \sum_{k} (\sqrt{\lambda_k} y_j^{\T} u_k )\varepsilon_{k} = \sum_{k} a_{k}^{(j)}\varepsilon_{k},
\]
where we denote $ a_{k}^{(j)} = \sqrt{\lambda_k} y_j^{\T} u_k $, and we also denote by $a^{(j)} \in \R^{\ell}$ the vector with corresponding coordinates. Observe that for $j \in B$, we have that 
\[
    \| a^{(j)}\|^{2} = \sum_{k} \lambda_{k} y_j^{\T} u_{k} u_{k}^{\T} y_{j} = y_{j}^{\T} M y_{j} \geq \tilde{r}^{2} .
\]
The Khintchin inequality due to \cite{szarek1976best} states that,
\[
    \frac{1}{\sqrt{2}} \| a^{(j)}\| \leq  \E \left| \sum_{k} a_{k}^{(j)} \varepsilon_{j} \right| \leq \| a^{(j)}\| \, .
\]
Furthermore, the lower tail of the bounded differences inequality (Theorem~{6.9} in \citet{boucheron2013concentration}) implies that
\[
    \P \left( \left| \sum_{k} a_{k}^{(j)} \varepsilon_{j} \right| < \frac{1}{\sqrt{2}} \| a^{(j)}\| - t   \right) \leq e^{-\frac{t^2}{2 (\| a^{(j)}\|^{2} + t\| a^{(j)}\|/3) }}.
\]
Taking $ t = \frac{1 - c}{\sqrt{2}} \| a^{(j)}\| $, we get that
\[
\P \left( \left| \sum_{k} a_{k}^{(j)} \varepsilon_{j} \right| \geq \frac{c}{\sqrt{2}} \| a^{(j)}\|  \right) \geq 1 - e^{-\frac{(1 - c)^2}{4(1 + (1-c)/(3\sqrt{2}))}},
\]
which for $c = \sqrt{2} / 5$ is greater than $0.1$. Hence, we can find a unit vector $v$ such that for at least one tenth of the indices $j \in B$,
\[
    | y_j^{\T} v| \geq \frac{1}{5} \tilde{r} = r  .
\]
One tenth of $B$ accounts for $0.1 \varepsilon = \kappa$, hence $ \theta$ cannot be an $(r, \kappa)$-combinatorial center.

Suppose that $\theta$ is a $(r, \varepsilon)$-spectral center. Again we will  prove that it is also a $(\tilde{r}, \kappa)$-combinatorial center by contradiction, with $\tilde{r} = \sqrt{(1-\varepsilon)/\varepsilon} \, r$, $\kappa = 2 \varepsilon$. Suppose it is not. Then, there is a unit vector $v$, such that for strictly more than $ \ell \kappa $ indices $j$, $|y_j^{\T} v| > \tilde{r}$. Denote this set of indices as $B$. Since $\theta$ is a spectral center, we get that for $M = vv^{\T}$,
\[
    \min_{w \in \Delta_{\ell, \varepsilon}} \sum_{j = 1}^{\ell} w_j  | y_j^{\T} v| ^{2} \leq r^2.
\]
The minimum puts {the} weight $1/ (\ell(1-\varepsilon))$ for $ \lfloor \ell(1 - \varepsilon) \rfloor $ indices $j$ with the smallest values $ |y_j^{\T} v| $. By the pigeonhole principle, strictly more than $\ell \kappa - \lceil \ell \varepsilon \rceil$ of them are in the set $B$. Hence, 
\[
    \min_{w \in \Delta_{\ell, \varepsilon}} \sum_{j = 1}^{\ell} w_j  | y_j^{\T} v| ^{2} > \frac{\kappa - \varepsilon}{1 - \varepsilon} \tilde{r}^{2} = \frac{\varepsilon}{1 - \varepsilon}  \tilde{r}^2 = r^{2}.
\]
This completes the proof by contradiction.

\section{Proof of Lemma~\ref{well_cond_convergence}}

Let us first calculate $w^{*}$ explicitly. Since $-M_{\lambda}(w; \Sigma, \mu)$ is strongly convex, and adding a Lagrangian multiplier $ -l(w^{\T}\mathbf{1} - 1) $ corresponding to the restriction $w^{\T}\mathbf{1} = 1$, we have that $w^*$ is the solution to
\[
    -\mu + \gamma \Sigma w - l \mathbf{1} = 0 \qquad \Rightarrow \qquad
    w = \frac{1}{\gamma} \Sigma^{-1} (\mu + l \mathbf{1})\,.
\]
Since $w^{\T} \mathbf{1} = 1$ we find that $ l = (\gamma - \mathbf{1}^{\T}\Sigma^{-1} \mu) / (\mathbf{1}^{\T} \Sigma^{-1} \mathbf{1}) $. Therefore,
\[
    w^{*} = \gamma^{-1} \Sigma^{-1} \mu + \frac{1 - \gamma^{-1} \mathbf{1}^{\T}\Sigma^{-1} \mu}{\mathbf{1}^{\T} \Sigma^{-1} \mathbf{1}} \Sigma^{-1}\mathbf{1}  \, .
\]

Denote $ \Pi_{0} = I - N^{-1} \mathbf{1}\mathbf{1}^{\T}$ the orthogonal projector onto the subspace of $\{ w: \; w^{\T} \mathbf{1} = 0 \} $, so that $ \Pi_{1}(x + y) = \Pi_{1}x + \Pi_{0} y $.  It is straightforward to check that $ \Pi_{0}(\gamma \Sigma w^{*} - \mu) $ vanishes, which is all we need to know about $w^{*}$ for the remaining of the proof.

Write $ \Delta(w) = \gamma {(}\Sigma w - \hat{a}(w){)}  - (\mu  - \hat{\mu}) $. Then, 
\begin{align*}
    w_{s + 1} - w^{*} &= w_{s} - w^{*} - \eta \Pi_0 [\gamma \Sigma w_s - \mu] + \eta \Pi_{0} \Delta(w_s) \\
    &= \left({I} - \eta\gamma \Pi_{0} \Sigma \right) (w_{s} - w^{*} ) + \eta \Pi_{0} \Delta(w_s) \\
    &= \Pi_{0}({I} -  \eta\gamma \Sigma)  (w_s - w^{*}) + \eta \Pi_{0} \Delta(w_s) \, ,
\end{align*}
{where since $ \eta \leq 1 / (\gamma\lambda_{\max}) $}, we have that ${I}- \eta\gamma \Sigma $ is positive definite, and $ \|{I} - \eta \gamma \Sigma\| = 1 - \eta \gamma \lambda_{\min} $. In {addition}, due to the requirement of the theorem for the estimators $\hat{\mu}$ and $\hat{a}(w)$, we have that
\[
    \| \Delta(w)\| \leq \Delta_{\mu} + \gamma \Delta_{\Sigma} \| w\| \, .
\]
Denoting $ \delta_s = \| w_s - w^{*}\| $, we have the recursive inequality,
\[
    \delta_{s + 1} \leq (1 - \eta \gamma \lambda_{\min}) \delta_{s} +  \eta \Delta_\mu + \eta\gamma \Delta_\Sigma \| w_{s}\| \, .
\]
We can link $ \| w_{s}\| $ to $\delta_s$ through a simple triangle inequality $ \| w_s\| \leq \| w^{*}\| + \delta_s $. We obtain,
\[
    \delta_{s + 1} \leq \fixaccepted{(}1 - \eta\gamma(\lambda_{\min} - \Delta_\Sigma)\fixaccepted{)} \delta_{s} + \eta \Delta_\mu +  \eta \gamma \Delta_{\Sigma} \| w^{*}\| \, .
\]
Denoting $ \kappa = 1 - \eta \gamma(\lambda_{\min} - \Delta_\Sigma) < 1 $ and $ x = \eta \Delta_\mu +  \eta \gamma \Delta_{\Sigma} \| w^{*}\| $, we expand our recursive inequality as follows,
\begin{align*}
    \delta_{s + 1} &\leq \kappa \delta_s + x \leq \kappa^2 \delta_{s-1} + \kappa x + x  \leq \kappa^{s + 1} \delta_{0} + (\kappa^{s} + \dots + 1) x \\
    & < \kappa^{s + 1} \delta_{0} + \frac{x}{1 - \kappa} \, .
\end{align*}
Substituting $\kappa$ and $x$ back, we obtain the result.

\subsection{Proof of {Theorem}~\ref{well_cond_main_res}}\label{proof_well_cond_main_res}

Simply substitute $ \Delta_{\mu} = C \| \Sigma \|^{1/2} \sqrt{\frac{\mathbf{r}(\Sigma) + \log(1/\delta)}{T}} $ and $ \Delta_{\Sigma} = C \| \Sigma \| \sqrt{\frac{\mathbf{r}(\Sigma) \log \mathbf{r}(\Sigma) + \log(1/\delta)}{T}} $ into Lemma~\ref{well_cond_convergence}, and take $ \eta = 1/(\gamma\lambda_{\max}) $. {Here $C > 0$ is an appropriately chosen absolute constant.} The condition \eqref{T_lower_bound} ensures that $ \Delta_{\Sigma} \leq \lambda_{\min}/2 $. We get that
\[
    \| w_{s} - w^{*} \| \leq \left( 1 - \frac{1}{2\kappa}\right)^{s} \| w_{0} - w^{*}\| + C' \frac{{\gamma^{-1}}\| \Sigma \|^{1/2} + \| \Sigma \| \| w^{*}\| }{\lambda_{\min}} \sqrt{\frac{\mathbf{r}(\Sigma) \log \mathbf{r}(\Sigma) + \log(1/\delta)}{T}}.
\]
Taking $ s \sim \log T $ steps, the first term will be dominated by second one. Furthermore, since the objective is quadratic and $w^{*}$ {is} its optimum, we have that
\begin{align*}
    M_{\gamma}(w^{*};\Sigma, \mu) - M_{\gamma}(w_{s};\Sigma, \mu) &= \frac{\gamma}{2} \| \Sigma^{1/2} (w_s - w^{*})\|^{2} \\
    & \lesssim 
    \frac{{\gamma^{-1}} \| \Sigma\|^{2} + {\gamma} \|\Sigma\|^{3} \|w^*\|^2}{\lambda_{\min}^{2}}\cdot \frac{\mathbf{r}(\Sigma)\log \mathbf{r}(\Sigma) + \log(1/\delta)}{T},
\end{align*}
hence the bound {follows}.

\section{Proof of Lemma~\ref{lemma_ill_cond}}
When $\Sigma$ is invertible, it is well known that the true minimum of the risk is $w^{*} = \Sigma^{-1} \mathbf{1} / (\mathbf{1}^{\T} \Sigma^{-1} \mathbf{1})$, so that $ \Pi_{0} \Sigma w^* $ vanishes. When $ \Sigma $ is not invertible, it is straightforward to check that $ \Pi_{0} \Sigma w^{*} = 0$ for any $ w^{*} \in \mathrm{Argmin}_{w^{\T} 1 = 1} R(w; \Sigma) $. Below we assume that $ w^{*} $ is some fixed optimal solution.

Write {$ \Delta(w) = \Pi_{0} [ \hat{a}(w) - \Sigma w] $}, so that the updates take the form $ w_{s + 1} = w_{s} - \eta \Pi_{0} \Sigma w_{s} - \eta \Delta(w_s)$. We first show that,
\begin{equation}\label{wsp1_ws_ineq}
    \| w_{s + 1} - w^{*} + \eta \Delta(w_{s}) \| \leq \| w_{s} - w^{*} \| \, .
\end{equation}
Using $ \Pi_{0}\Sigma w^{*} = 0$, we have
\begin{align*}
    w_{s + 1} - w^{*} + \eta \Delta(w_s)  &=  w_{s} - w^{*} - \eta \Pi_0 \Sigma w_s \\
    &=  w_{s} - w^{*} + \eta \Pi_{0} \Sigma w^{*} - \eta \Pi_0 \Sigma w_s \\
    & = (I - \eta \Pi_{0} \Sigma \Pi_{0}) (w_s - w^{*}) ,
\end{align*}
where for the last line we used the fact that $ (w_s - w^{*}) = \Pi_{0} (w_s - w^{*}) $, since both $w^*$ and $w_s$ sum up to one.
Using the condition $ \eta \leq 1/\| \Sigma \| $, the matrix $I - \eta \Pi_{0} \Sigma \Pi_{0}$ is a contraction, so the inequality \eqref{wsp1_ws_ineq} follows.

Applying further the triangle inequality, we have that
\begin{align*}
    \| w_{s + 1} - w^{*} \|  &\leq \| w_{s} - w^{*} \| + \| \Delta(w_s) \| \leq (1 + \eta \Delta_{\Sigma}) \| w_s - w^{*} \| + \eta \Delta_{\Sigma} \| w^{*} \| \\
    & \leq \dots \\
    & \leq (1 + \eta \Delta_{\Sigma})^{s+1} \left( \| w_0 - w^{*}\| + s \eta \Delta_{\Sigma} \| w^*\| \right).
\end{align*}
Let $n$ be an integer and
assume that $ n \eta \Delta_{\Sigma} \leq 1 $. Then, using the inequality $ (1 + 1/n)^n \leq e$, for each $s = 0, 1, \dots, n$, {we have}
\begin{equation}\label{wtws_bounded}
    \max(\| w_s\|, \| w_{s} - w^{*} \|) \leq (e + 1) \left( \| w_0 - w^{*}\| + \| w^{*} \| \right) \myeq M \, .
\end{equation}
Further, we apply a standard trick {in optimization;} see e.g., Theorem~{3.5} in \citet{bubeck2014convex}. Let us denote $ R^{*}(w) = R(w; \Sigma)$, which is a convex and $ \| \Sigma \| $-smooth function. Therefore, it holds that for any $u, w$,
\begin{equation}\label{convex_smooth}
    0 \leq R^{*}(u) - R^{*}(w) - \nabla R^{*}(w) (u - w) \leq \frac{\| \Sigma \|}{2} \| u - w \|^{2} \, .
\end{equation}
Applying this inequality for $w_s$ and $ w_{s + 1} = w_{s} - \eta \Pi_{0} \Sigma w_{s} - \eta \Delta(w_s) $, we first obtain that
\begin{align*}
    R^{*}(w_{s + 1}) - R^{*}(w_s) \leq & -\eta (\Sigma w_{s})^{\T} \left[ \Pi_{0} \Sigma w_{s} + \Delta(w_s)  \right] + \frac{\eta^{2} \| \Sigma \|}{2} \left \| \Pi_{0} \Sigma w_{s} + \Delta(w_s)  \right\|^{2} \\
    \leq & -\eta (\Sigma w_{s})^{\T} \left[ \Pi_{0} \Sigma w_{s} + \Delta(w_s)  \right] + \frac{\eta}{2} \left \| \Pi_{0} \Sigma w_{s} + \Delta(w_s)  \right\|^{2} \\
    = & - \frac{\eta}{2} \| \Pi_{0} \Sigma w_{s} \|^{2}  + \frac{\eta}{2} \| \Delta(w_s) \|^{2} .
\end{align*}
Observe that due to \eqref{convex_smooth},
\[
    R^*(w_s) - R^*(w^*) \leq (\Sigma w_s)^{\T}(w_s - w^*) \leq \| \Pi_0 \Sigma w_s \| \| w_s - w^{*}\| \leq M \| \Pi_0 \Sigma w_s \|,
\]
where in the last inequality we also use the bound \eqref{wtws_bounded}. Furthermore, $ \| \Delta(w_s)\| \leq \Delta_{\Sigma} M $. Denoting $ \delta_{s} = R^{*}(w_s) - R^{*}(w^*)$, we obtain the recursive inequality,
\[
    \delta_{s + 1} \leq \delta_{s} -\frac{\eta}{4M^{2}} \delta_s^{2} + \eta \Delta_{\Sigma}^{2} M^{2} \, .
\]
Denoting additionally $ \alpha_{s} = \max\{0, \delta_{s} - s \eta \Delta_{\Sigma}^{2} M^{2}\} $, we can easily derive that $\alpha_{s + 1}  \leq \max\left\{ 0, \alpha_{s} - \frac{\eta}{4M^2} \alpha_{s} \right\} $. It is straightforward to check that $ \alpha_{0} \leq \|\Sigma\| M^2 \leq \frac{4M^2}{\eta}$ and $ \alpha_{s + 1} \leq \alpha_{s} $. Therefore, we conclude that we can drop the positive part, so that $ \frac{1}{\alpha_{s}} \leq \frac{1}{\alpha_{s + 1}} - \frac{\eta}{4M^2} \frac{\alpha_s}{\alpha_{s + 1}} \leq \frac{1}{\alpha_{t + 1}} - \frac{\eta}{4M^2} $. Hence, the bound $ \frac{1}{\alpha_{t}} \geq \frac{\eta}{4M^2} t $ follows . Therefore, the following inequality holds
\[
    R^{*}(w_t) - R^{*}(w^*) = \delta_{t} \leq \alpha_{t} + t \eta \Delta_{\Sigma}^{2} M^{2} \leq \frac{4M^2}{\eta t} + t \eta \Delta_{\Sigma}^{2} M^{2} \, ,
\]
which completes the proof.

\section{Weights visualization}\label{sec:vizu}
\begin{figure}[H]
    \centering
    \includegraphics[width=1.0\textwidth]{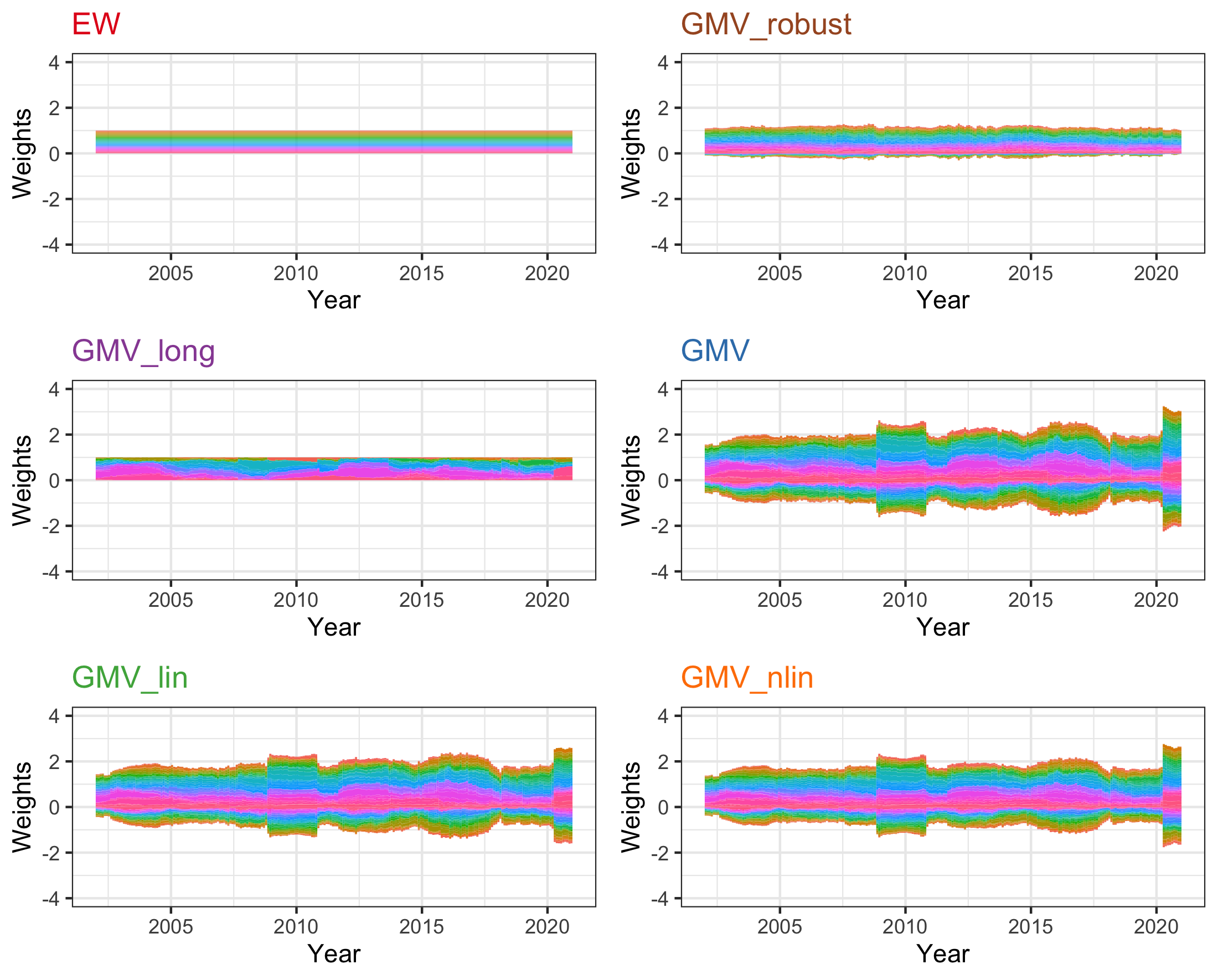}
    \caption{Weights of assets in GMV portfolios, 81 S\&P100 constituents, 20020101 - 20201231}
    \label{Weights_SP}
      \quantlet{RobustM_PerformanceSP100}{RobustM\_PerformanceSP100}
\end{figure}

 \begin{figure}[H]
    \centering
    \includegraphics[width=1.0\textwidth]{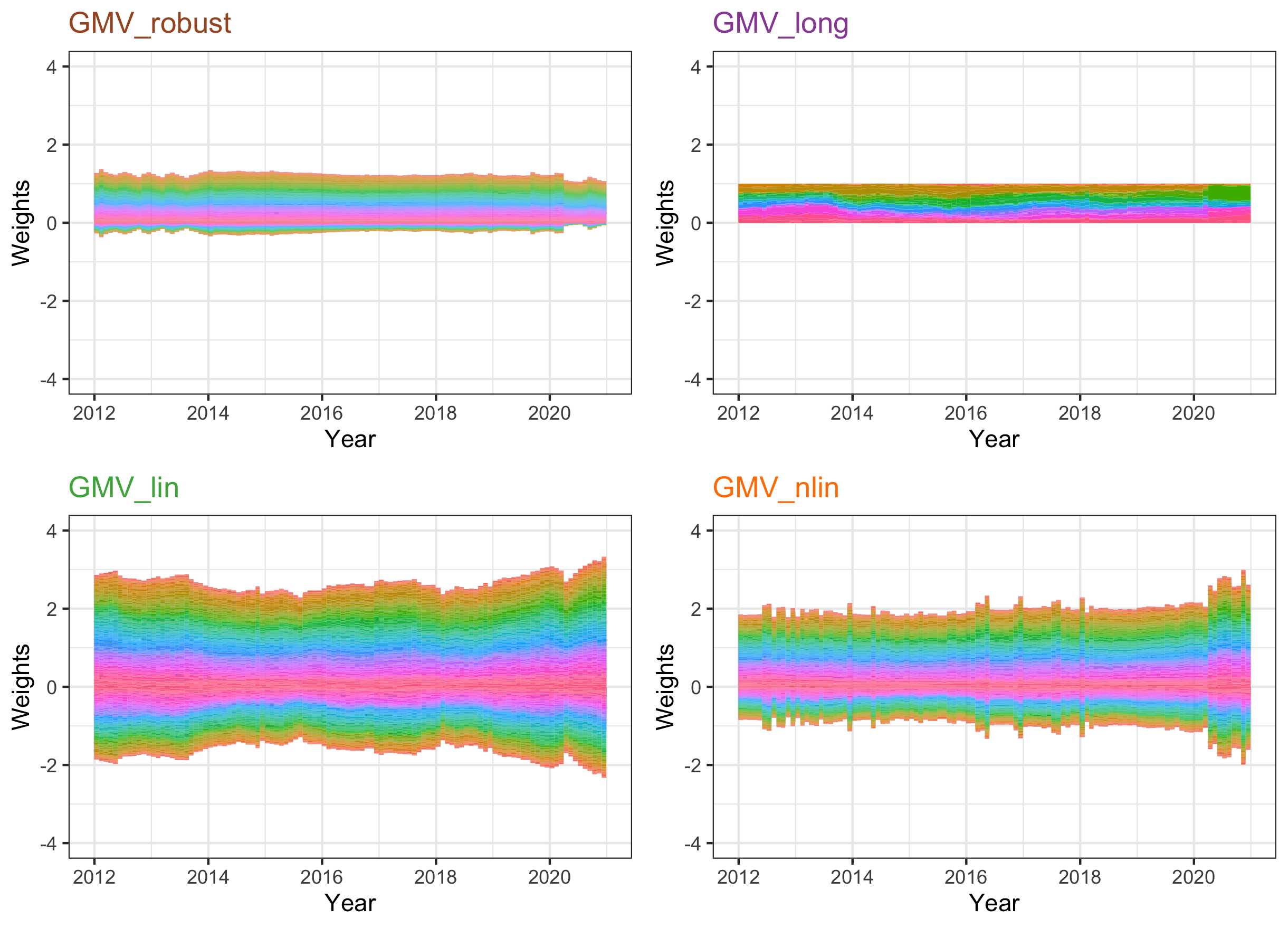}
   \caption{Weights of assets in GMV portfolios, 600 Russell3000 constituents, 20120101 - 20201231}
    \label{Weights_Russell}
     \quantlet{RobustM_PerformanceRussell3000}{RobustM\_PerformanceRussell3000}
\end{figure}
\section{Performance results with window length \texorpdfstring{$T = 252$}{Lg}}
\label{sec:perf}
\tab{Performance_SP1001}{\footnotesize 
{ 
Out-of-sample performance measures of benchmark portfolios, 81 stocks of S\&P100, monthly rebalancing: TTO, target turnover; TO, Turnover;  CW, cumulative Wealth; SD, standard deviation; SR, Sharpe ratio; CR, Calmar ratio. All estimates are obtained from daily values (5031 out-of-sample returns) for both gross and net returns. Time period: 20010101 - 20201231.The difference tests for the turnovers are obtained from a sample $t$-test. The difference test on the Sharpe ratio and Variance used the approach of \cite{ledoit2008robust} and \cite{ledoit2011robust}, for which the R code is available at M. Wolf's website.
$p$-values are reported in parentheses with respect to the GMV\_{robust} portfolio.}
\quantlet{RobustM_PerformanceSP100}{RobustM\_PerformanceSP100}}

\tab{Performance_Russell30001}{\footnotesize 
{Out-of-sample performance measures of benchmark portfolios, 600 stocks of Russell3000, monthly rebalancing and : TTO, target turnover; TO, Turnover, TC, transaction costs;  CW, cumulative Wealth; SD, standard deviation; SR, Sharpe ratio; CR, Calmar ratio. All estimates are obtained from daily values (2515 out-of-sample returns) for both gross and net returns. Time period: 20110101 - 20201231. The difference tests for  turnovers are obtained from a sample $t$-test. The difference test on the Sharpe ratio and Variance used the approach of \cite{ledoit2008robust} and \cite{ledoit2011robust}, for which the R code is available at  M. Wolf's website. $p$-values are reported in parentheses with respect to the GMV\_{robust} portfolio.} 
\quantlet{RobustM_PerformanceRussell3000}{RobustM\_PerformanceRussell3000}}

\end{document}